\documentclass[11pt]{article}

\usepackage[a4paper, vmargin=3cm,hmargin=3.2cm,nohead]{geometry}

\usepackage{hyperref}
\usepackage{amsmath}
\usepackage{amssymb}
\usepackage{amsthm}
\usepackage{mathtools}
\usepackage{graphicx}
\usepackage{color}
\usepackage{xcolor}
\usepackage{csquotes}

\usepackage[ruled,vlined,linesnumbered]{algorithm2e}
\usepackage{tikz}
\usetikzlibrary{graphs, graphs.standard, calc, arrows.meta, bending, shapes.geometric}
\usepackage{nicefrac}
\usepackage{cleveref}
\usepackage{thm-restate}
\usepackage{subcaption}
\usepackage{array}

\usepackage{bm}
\renewcommand{\vec}[1]{\bm{#1}}

\theoremstyle{plain} 
\newtheorem{theorem}{Theorem}[section]
\newtheorem{lemma}[theorem]{Lemma}     

\theoremstyle{definition} 
\newtheorem{definition}[theorem]{Definition}
\newtheorem{remark}[theorem]{Remark}

\newcommand{\set}[1]{\ensuremath{\{#1\}}}

\newcommand{\mech}{\ensuremath{\mathcal{M}}}
\newcommand{\opt}{\ensuremath{\mathrm{OPT}}}

\newcommand{\R}{\ensuremath{\mathbb{R}}}

\newcommand{\E}{\mathbb{E}}
\newcommand{\pe}[1]{\E[#1]}

\newcommand{\SC}{\ensuremath{\mathrm{MC}}}

\title{Improved Lower Bounds and Output Augmentation for Facility Location Mechanisms}

\author{Rafael Gomes\thanks{Centrum Wiskunde \& Informatica (CWI), Amsterdam, The Netherlands} \and Sophie Klumper\footnotemark[1] \and Guido Sch\"afer\thanks{Centrum Wiskunde \& Informatica (CWI) and University of Amsterdam, The Netherlands} \and Jens Schl\"oter\footnotemark[1]}

\date{}

\begin{document}

\maketitle

\begin{abstract}

We study the strategic facility location problem under the egalitarian objective, where a mechanism uses the reported locations of a set of agents in Euclidean space to select a facility location that minimizes the maximum distance to any agent. We restrict our attention to strategyproof mechanisms, ensuring that no agent can benefit from misreporting their location.

As our main results, we prove an asymptotic lower bound of $1 + \sqrt{d/(2(d+1))}$ on the approximation ratio of any mechanism that is strategyproof in expectation in $\R^d$. We show that this barrier is driven by large populations by providing a randomized $\sqrt{2}$-approximate mechanism for the two-agent case.

We then consider an output-augmented framework, which allows the facility to be placed outside the agents' restricted domain. For the setting where agents are restricted to a line but the facility can be anywhere in the plane, we design a deterministic strategyproof $\sqrt{2}$-approximate mechanism with a matching lower bound, showing that output augmentation can replace the need for randomness. For the setting where the agents' reports lie on the unit circle but the facility can be placed anywhere in $\mathbb{R}^2$ we introduce a randomized $3/2$-approximate mechanism that is group-strategyproof in expectation.

\end{abstract}

\section{Introduction}

The facility location problem has been a prototypical benchmark for approximate mechanism design without money since its introduction by Procaccia and Tennenholtz~\cite{procaccia_approximate_2013}. 
A planner runs a mechanism $\mech$ to determine a (random) location $Y$ of a facility to serve a set $N = \set{1, \dots, n}$ of $n \ge 2$ agents, where each agent $i \in N$ has a private location $x_i \in \R^d$ with $d \ge 1$. 
Throughout this work, we consider distances in Euclidean space and adopt the egalitarian objective:
Given locations $\vec{x} = (x_1, \dots, x_n)$, the cost of agent $i$ is the (expected) Euclidean distance $\pe{d(x_i, Y)}$ from $x_i$ to $Y = \mech(\vec{x})$. 
Under the \emph{egalitarian} objective, the mechanism seeks to minimize the maximum distance experienced by any agent, i.e., $\SC(\mech, \vec{x}) = \pe{\max_i d(x_i, \mech(\vec{x}))}$.
Because agents may misreport their locations to pull the facility closer to their true locations, we are interested in designing mechanisms that incentivize truthfulness: 
$\mech$ is \emph{strategyproof in expectation} if no agent can decrease their expected distance by unilaterally misreporting, and \emph{group-strategyproof in expectation} if no coalition of agents can jointly misreport so that every member decreases their expected distance.

On the real line $\R$, this problem is essentially settled: Procaccia and Tennenholtz~\cite{procaccia_approximate_2013} established that the optimal egalitarian approximation ratio is $2$ for deterministic mechanisms and $\nicefrac{3}{2}$ for randomized mechanisms that are strategyproof in expectation. In higher-dimensional Euclidean space, the picture is far less complete for randomized mechanisms, and the gap between known upper and lower bounds on the approximation guarantee is striking. 
The best-known randomized mechanism in $\mathbb{R}^d$ is the \emph{Centroid Mechanism} by Tang et al.~\cite{tang_characterization_2020} achieving an approximation ratio of $2 - \nicefrac{1}{n}$. 
On the lower bound side, the strongest result prior to this work was a bound of $1.118$ in the plane $\R^2$, due to Balkanski et al.~\cite{balkanski_randomized_2024}. Closing this gap, even in $\mathbb{R}^2$, has remained open. 

A second open direction that we address in this paper concerns the design of mechanisms whose output is permitted to lie outside the agents' input domain. 
Note that this is conceptually different from resource augmentation in the sense of~\cite{Roughgarden20a}: we do not weaken the benchmark, we enrich the mechanism's output space. 
Formally, we distinguish between the input space $\mathcal{I}$ where agents' reports live and the output space $\mathcal{O}$ where the facility is placed, with $\mathcal{I} \subseteq \mathcal{O}$. 
The standard setting in the literature takes $\mathcal{I} = \mathcal{O} = \mathbb{R}^d$; in the \emph{output-augmented setting} we allow $\mathcal{I} \subsetneq \mathcal{O}$. 
Given reports $\vec{x} = (x_1, \dots, x_n) \in \mathcal{I}^n$, the optimum $\opt(\vec{x})$ is the radius of the smallest enclosing ball of $\vec{x}$ in the \emph{output space} $\mathcal{O}$. 
The \emph{approximation ratio} of $\mech$ is $\sup_{\vec{x} \in \mathcal{I}^n} \SC(\mech, \vec{x})/\opt(\vec{x})$. 

Beyond its theoretical interest, this relaxation is also practically motivated: A radio antenna serves subscribers distributed along a road, but need not stand on the road itself. An offshore wind platform serves towns arranged along a coastal arc, but can be built at sea. 
Whether output augmentation allows mechanisms to bypass classical lower bounds is the main question we study in this setting.

\paragraph{\bfseries Our Contributions.}

Our results address two settings:

\medskip\noindent
\emph{(1) Standard setting ($\mathcal{I} = \mathcal{O} = \mathbb{R}^d$).}\quad
We prove that, as $n \to \infty$, every randomized mechanism that is strategyproof in expectation for egalitarian facility location in $\mathbb{R}^d$ has approximation ratio at least $1 + \sqrt{\nicefrac{d}{2(d+1)}}$. 
For the planar case $d = 2$, this gives a bound of $1+{\nicefrac{1}{\sqrt{3}}} \approx 1.577$, improving the previous best lower bound of $1.118$~\cite{balkanski_randomized_2024}; the bound tends to $1 + \nicefrac{1}{\sqrt{2}} \approx 1.707$ as $d \to \infty$. The construction places agents at the vertices of a regular simplex inscribed in the unit sphere and then has one cluster of agents deviate to the boundary of a larger sphere centered at their original position. 
The new bound makes essential use of infinitely large populations, but in $\R^2$ a discretized argument gives a lower bound that grows with $n$, exceeding $\nicefrac{3}{2}$ already for $n \ge 15$.
For small $n$ we obtain complementary results: a lower bound of $1.277$ for $n \ge 3$, and a randomized $\sqrt{2}$-approximation for $n = 2$, which beats the $\nicefrac{3}{2}$ ratio achievable on the line.

\medskip\noindent
\emph{(2) Output-augmented setting ($\mathcal{I} \subsetneq \mathcal{O}$).}\quad
We address the question whether output augmentation yields genuinely more powerful strategyproof mechanisms.

For agents on a line $\mathcal{I} = \R$ and facility in the plane $\mathcal{O} = \mathbb{R}^2$, we design a simple \emph{deterministic} strategyproof mechanism with approximation ratio $\sqrt{2}$, matched by a $\sqrt{2}$ lower bound for any deterministic mechanism.
Notably, this shows that deterministic mechanisms even beat the classical \emph{randomized} lower bound of $\nicefrac{3}{2}$ for the line without augmentation~\cite{procaccia_approximate_2013}, demonstrating that the augmented framework is strictly more powerful. We generalize this mechanism to the setting where $\mathcal{I} = \R^d$ and  $\mathcal{O} = \mathbb{R}^{d+1}$, and show an approximation ratio of $\sqrt{d+1}$. For $d=2$, this improves upon the best-possible approximation factor of $2$~for deterministic mechanisms~\cite{ChanLW26,GoelH23} in two dimensions, i.e., $\mathcal{I} = \mathcal{O} = \mathbb{R}^2$. 

Our main technical contribution concerns agents on the unit circle $\mathcal{I} = S^1$ and facility in the plane $\mathcal{O} = \mathbb{R}^2$. We design a randomized mechanism, the \emph{Chord-Midpoint Mechanism}, that is group-strategyproof in expectation with approximation ratio $3/2$. The mechanism identifies two extreme agents $A$ and $B$ that delimit the minimal arc containing all reports, and outputs $A$, $B$, or the midpoint of the chord $AB$ according to a probability $\lambda(\alpha)$ that depends on the arc's angular span. The choice of $\lambda(\alpha)$ is delicate: it must be small enough that agents have an incentive to reveal their true positions, and large enough that the realized facility remains close to the chord midpoint in expectation. 
We complement this with a lower bound of $2$ for any deterministic, unanimous, group-strategyproof mechanism in the same setting, showing that randomization is necessary to break the barrier.

To the best of our knowledge, studying output-augmented settings in this form is new, particularly when non-trivial input topologies are considered such as $\mathcal{I} = S^1$ and $\mathcal{O} = \mathbb{R}^2$ (see Related Work for prior work on dimension augmentation). 
We see the study of such concrete topologies as a stepping stone toward designing optimal mechanisms for increasingly rich input domains, which might ultimately lead to a resolution of the general setting.

\paragraph{\bfseries Techniques.}

In our lower bounds for the standard setting, we combine symmetric simplex configurations with a cluster-deviation lemma (showing that a coalition of co-located agents cannot reduce their expected distance to their true position) to reduce the analysis to a purely geometric statement about the maximum diameter of a set under a fixed circumradius. 
A consequence of Jung's theorem is that this strategy is tight for the simplex-based construction, so any further improvement will require a structurally different approach.  
Given this geometric connection, it remains an intriguing open question whether a matching mechanism exists. 

In the setting with reports on the circle $\mathcal{I} = S^1$ and facility in the plane $\mathcal{O} = \mathbb{R}^2$, the report-dependent mixing parameter $\lambda(\alpha)$ allows the mechanism to interpolate smoothly between the deterministic optimum (when the agents fill a semicircle) and a randomization that resembles the classical $3/2$-mechanism on the line (when the arc is small). The analysis hinges on a tight factorization of the distance from any agent to the moving chord midpoint (established in Lemma~\ref{lem:distance_factorization}). 
Our proof that the mechanism is group-strategyproof in expectation follows from a decomposition argument that reduces coalitional deviations to a sequence of unilateral arc-expansion and arc-shrinking moves.

\paragraph{\bfseries Related Work.}

Following~\cite{procaccia_approximate_2013}, 
extensive research has investigated truthful mechanisms across various settings. While the literature spans multiple social cost objectives, we focus on the egalitarian objective in this paper. 
We give a brief overview of further related work and refer to \cite{chan_mechanism_2021} for a more extensive survey.

\emph{Facility Location in Higher Dimensions.}  
Moulin \cite{moulin_strategy-proofness_1980} showed that the generalized median scheme characterizes all deterministic strategyproof mechanisms on the line $\R$; subsequently, this result was extended to higher dimensions \cite{border_straightforward_1983,kim_nonmanipulability_1984,peters_range_1993,tang_characterization_2020}. 
The characterization of mechanisms that are strategyproof in expectation remains an intriguing open problem.
Under the stronger requirement of group-strategyproofness, Tang et al. \cite{tang_characterization_2020} characterized randomized, translation-invariant mechanisms in strictly convex spaces as 2-dictatorial rules.
In higher dimensions, the optimal \emph{utilitarian} approximation ratio for randomized mechanisms has also remained an open problem, with recent progress made in \cite{barak2026facilitylocationmechanismdesign} with, among other results, a $\nicefrac{4}{\pi}$ approximation in $\mathbb{R}^2$ that strictly improves upon the achievable guarantee of deterministic mechanism.

\emph{Facility Location on Networks.} Another line of work embeds the problem on graphs, where agents and facilities both reside along the edges (in contrast to the output-augmented setting of this paper).
Schummer and Vohra \cite{schummer_strategy-proof_2002} provided characterizations of deterministic mechanisms showing that strategyproof rules behave similarly to generalized medians on trees but reduce to local dictatorships when restricted to cycles. In the randomized setting, Alon et al.~\cite{alon_strategyproof_2010} analyzed approximation bounds under the minimax objective, proving a $2 - o(1)$ lower bound for trees and presenting a tight $3/2$ approximation mechanism for the circle graph under the egalitarian cost.

\emph{Constrained and Augmented Spaces.} There is a body of literature studying settings where the agent input space and the facility output space do not coincide. In \emph{constrained facility location}, the output space is a strict subset of the input space, i.e., $\mathcal{O} \subsetneq \mathcal{I}$, orthogonal to the output-augmented setting we consider in this paper. 
A variety of papers have studied this variant to understand what happens when the facility is restricted to specific sub-regions or discrete candidate locations \cite{barbera_voting_1997,feldman_voting_2016,li_locating_2025, tang_mechanism_2020}. 
Conversely, in \emph{dimension-augmented} settings, the output space strictly encompasses the input space. To the best of our knowledge, the only prior work to explore this paradigm is by Fullerton et al.~\cite{fullerton_constant-approximate_2026}, who investigate the two-facility location problem under the utilitarian objective with $\mathcal{I} = \R$ and $\mathcal{O} = \R^2$.

\section{Preliminaries}

Let $\vec{x} \in \mathcal{I}^n$ be a profile of reported locations. We use the standard notation $\vec{x}_{-i}$ to denote the reports of all agents except agent $i$, and similarly $\vec{x}_{-S}$ for the reports of all agents outside a coalition $S \subseteq N$. Throughout the paper, we use $d(x, y)$ and $\|x - y\|$ interchangeably to denote the Euclidean distance between two points $x$ and $y$. 

A (randomized) mechanism $\mech$ takes a profile $\vec{x} \in \mathcal{I}^n$ as input and outputs a (random) location $Y = \mech(\vec{x}) \in \Delta(\mathcal{O})$.\footnote{We use $\Delta(\mathcal{O})$ to refer to the set of all distributions over $\mathcal{O}$.} $\mech$ is deterministic if it outputs a location $y = \mech(\vec{x})$ with probability $1$ for each $\vec{x}$.
A deterministic mechanism $\mech$ is \emph{strategyproof} if for any profile $\vec{x}$, agent $i$, and unilateral deviation $x'_i \in \mathcal{I}$, we have $d(x_i, \mech(\vec{x})) \le d(x_i, \mech(x'_i, \vec{x}_{-i}))$; similarly, a randomized mechanism $\mech$ is \emph{strategyproof in expectation} if $\E[d(x_i, \mech(\vec{x}))] \le \E[d(x_i, \mech(x'_i, \vec{x}_{-i}))]$. Furthermore, $\mech$ satisfies \emph{group-strategyproofness in expectation} if no coalition $S \subseteq N$ can jointly misreport to a partial profile $\vec{x}'_S$ such that $\E[d(x_i, \mech(\vec{x}'_S, \vec{x}_{-S}))] < \E[d(x_i, \mech(\vec{x}))]$ for all $i \in S$.

We say that a mechanism $\mech$ is \emph{anonymous} if its outcome is invariant to permutations of the agents' reports. $\mech$ is \emph{unanimous} if, whenever all agents $i \in N$ report the same location $x_i = u \in \mathcal{I}$, the facility is placed at $u$.

A mechanism $\mech$ is \emph{$\alpha$-approximate} if for any $\vec{x} \in \mathcal{I}^n: \SC(\mech, \vec{x}) \leq \alpha \cdot \opt(\vec{x})$.

\section{Randomized Mechanisms}
\label{sec:randmech}

\subsection{Lower Bound in $\mathbb{R}^d$}

In this section, we consider the standard setting with $\mathcal{I} = \mathcal{O} = \R^d$, where $d \ge $1.

\begin{restatable}{theorem}{LowerBoundTheorem} \label{lowerboundtheorem}
Any randomized strategyproof mechanism for facility location in $\R^d$, with $d \ge 1$, has an egalitarian approximation ratio of at least $1 + \sqrt{\nicefrac{d}{2(d+1)}}$.
\end{restatable}

The proof of the theorem will rely on an initial profile where the agents are distributed equally among the vertices of a regular $d$-simplex in $\mathbb{R}^d$, with multiple agents located in each vertex. We denote this set of vertices, centered at the origin $O$, as
$$S = \{x_0, \dots, x_d\}$$
where the locations satisfy
$$d(O, x_k) = 1 \quad \forall x_k \in S.$$

We begin with a lemma establishing that, in expectation, the output of a strategyproof mechanism cannot be arbitrarily close to every cluster simultaneously.

\begin{lemma}\label{lem:expected_cost_simplex}
Given the profile $\vec{x}$ where agents are located at the vertices of $S$, for any randomized mechanism $Y = \mech(\vec{x})$ there exists a vertex $x_j \in S$ such that $\pe{d(Y, x_j)} \ge 1$. 
\end{lemma}

\begin{proof}
For any realized facility location $y \in \mathbb{R}^d$, the sum of the Euclidean distances to the vertices of $S$ is minimized at the geometric median, which for a regular simplex coincides with its circumcenter $O$. Because the distance from $O$ to each of the $d+1$ vertices is $1$, the minimum possible sum of distances is $d+1$. Thus, for any $y \in \mathbb{R}^d$, we have the bound:
\begin{equation} \label{eq:median_bound}
    \sum_{k=0}^d d(y, x_k) \ge \sum_{k=0}^d d(O, x_k) = d+1.
\end{equation}
Taking the expectation over the randomness of the mechanism $\mech$, by linearity of expectation we obtain:
\begin{equation*}
    \sum_{k=0}^d \E\left[d(Y, x_k)\right] = \E\left[ \sum_{k=0}^d d(Y, x_k)\right] \ge d+1.
\end{equation*}
Because the sum of these $d+1$ expected distances is bounded below by $d+1$, at least one term in the sum must be at least 1. Therefore, there exists an index $j \in \{0, \dots, d\}$ such that $\E[d(Y, x_j)] \ge 1$.

\end{proof}

The proof of Theorem \ref{lowerboundtheorem} will also involve moving a group of co-located agents all at once. At first glance, this might seem problematic because strategyproofness only prevents unilateral deviations. However, the following lemma shows that if a group of agents sharing the same true location misreport at the same time, they still cannot decrease their expected distance.

\begin{lemma}\label{lem_cluster_deviation}
Let $\mech$ be a randomized mechanism that is strategyproof in expectation. Let $\vec{x}$ be a profile where a subset of $k$ agents share the same true location $u$. 
Consider a profile $\vec{x}'$ where all $k$ agents misreport their locations. 
Then $\E[d(\mech(\vec{x}'), u)] \ge \E[d(\mech(\vec{x}), u)].$ 
\end{lemma}
\begin{proof}
We construct a sequence of profiles $\vec{x} = \vec{x}^0, \vec{x}^1, \dots, \vec{x}^k = \vec{x}',$ where for each $i \in \{1, \dots, k\}$, $\vec{x}^i$ is the profile obtained after $i$ agents have changed their report from $u$ to their respective misreported location in $\vec{x}'.$ 

Consider the transition from profile $\vec{x}^{i-1}$ to $\vec{x}^i.$ The only difference between these two profiles is that a single agent, whose true location is $u$, changes their reported location. Because the mechanism $\mech$ is strategyproof in expectation, this agent cannot decrease their expected distance to their true location $u$ by misreporting, i.e., $\E[d(\mech(\vec{x}^i), u)] \ge \E[d(\mech(\vec{x}^{i-1}), u)].$ 
Chaining these inequalities for all $i \in \{1, 2, \dots, k\}$, we obtain:
\[
\E[d(\mech(\vec{x}'), u)] = \E[d(\mech(\vec{x}^k), u)] \ge \dots \ge \E[d(\mech(\vec{x}^0), u)] = \E[d(\mech(\vec{x}), u)].
\]
\end{proof}

We are now ready to prove the main result.

\LowerBoundTheorem*

\begin{proof}
    Consider the unit sphere $C_O = S^{d-1}$ in $\R^d$ centered at the origin $O = (0, \dots, 0)$. Let $S$ be a regular simplex inscribed in $S^{d-1}$, defined by $d+1$ vertices $x_0, x_1, \dots, x_d$, where we orient the simplex such that $x_0 = (1, 0, \dots, 0)$. Because the circumradius of $S$ is $1$, the distance from the origin to any vertex is $d(O, x_k) = 1$ for all $k \in \{0, 1, \dots, d\}$. The uniform side length of this regular simplex is given by $a_d = \sqrt{2 + \nicefrac{2}{d}}$. Therefore, the pairwise distance is $d(x_j, x_k) = a_d$ for all $j \neq k$.
    See Figure~\ref{fig:lb-configuration} for an illustration of the construction for $d = 2$.

    \begin{figure}[h!]
\centering
\begin{tikzpicture}[scale=1.1]
    \coordinate (o) at (0,0);
    \coordinate (x1) at (1,0);
    \coordinate (x2) at (-0.5,0.8660254038);
    \coordinate (x3) at (-0.5,-0.8660254038);
    \coordinate (x4) at (2.7320508,0);

    \draw[->] (-1.45,0) -- (3.25,0) node[above right] {$x$};
    \draw[->] (0,-2.2) -- (0,2.2) node[above left] {$y$};

    \draw[blue, thick] (o) circle (1);
    \draw[red, thick] (x1) circle (1.7320508076);
    
    \foreach \x in {-1} {
        \draw[very thick] (\x,0.04) -- (\x,-0.04) node[below left] {\small $\x$};
    }
    \foreach \x in {1} {
        \draw[very thick] (\x,0.04) -- (\x,-0.04) node[below right] {\small $\x$};
    }
    \draw[very thick] (2.7320508,0.04) -- (2.7320508,-0.04) node[below right] {\small $1+\sqrt{3}$};
    
    \foreach \y in {-1} {
        \draw[very thick] (0.04,\y) -- (-0.04,\y) node[below right] {\small $\y$};
    }
    \foreach \y in {1} {
        \draw[very thick] (0.04,\y) -- (-0.04,\y) node[above right] {\small $\y$};
    }    

    \draw (x1) -- (x2) -- (x3) -- (x1);

    \fill[black] (o) circle (1.2pt) node[above left] {$O$};
    \fill[blue] (x1) circle (1.5pt) node[above right] {$x_0$};
    \fill[blue] (x2) circle (1.5pt) node[above left] {$x_1$};
    \fill[blue] (x3) circle (1.5pt) node[below left] {$x_2$};
    \fill[red] (x4) circle (1.5pt) node[above right] {$x'_0$};

    \node[blue] at (-1.15,0.5) {$C_O$};
    \node[red] at (2.25,1.5) {$C_{x_0}$};

    \node at (0.3,0.6) {\small $\sqrt3$};
\end{tikzpicture}
\caption{Lower bound instance for $d=2$.}
\label{fig:lb-configuration}
\end{figure}
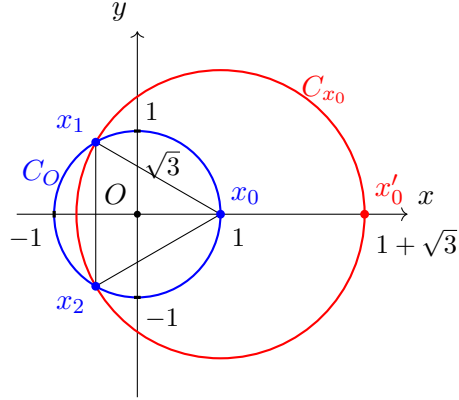

    Let $n$ be a multiple of $d+1$. Consider an initial profile $\vec{x}$ of $n$ agents equally distributed among the $d+1$ vertices of $S$. Let $Y = \mech(\vec{x})$ denote the random facility location chosen by the mechanism. 
    Applying Lemma~\ref{lem:expected_cost_simplex} to the vertices of the regular simplex $S$, the expected distance from the facility to at least one of these vertices is at least $1$.  
    Without loss of generality, assume this holds for $x_0$, i.e.,
    \begin{equation} \label{eq:lower_bound_x0}
        \E\left[d(Y, x_0)\right] \ge 1.
    \end{equation}

    We now construct a sequence of modified profiles $\vec{x}'$ as $n \to \infty$. Let $C_{x_0}$ be the $(d-1)$-sphere centered at $x_0$ with radius $a_d$. Observe that all other vertices $x_1, \dots, x_d$ inherently lie on $C_{x_0}$ because $d(x_0, x_k) = a_d$. We modify the profile by having the agents originally located at $x_0$ uniformly deviate over the surface of $C_{x_0}$.

    Let $Y' = \mathcal{M}(\vec{x}')$. By Lemma \ref{lem_cluster_deviation}, a simultaneous deviation by the subset of agents from their true location $x_0$ to the boundary of $C_{x_0}$ cannot decrease their expected distance to $x_0$. Applying \eqref{eq:lower_bound_x0}, we obtain:
    \begin{equation*}
        \E[d(Y', x_0)] \geq \E[d(Y, x_0)] \geq 1.
    \end{equation*}

    To evaluate the egalitarian cost of $\mech$ on $\vec{x}'$, consider any realized facility location $y'$. The supremum distance from $y'$ to any point on the sphere $C_{x_0}$ is achieved by extending the line segment from $y'$ through the center $x_0$ to the far boundary. This supremum is exactly $d(y', x_0) + a_d$. As $n \to \infty$, the discretely distributed deviating agents densely populate $C_{x_0}$, and the maximum distance to any agent in the profile converges to this supremum:
    \begin{equation*}
        \lim_{n \to \infty} \max_{i \in [n]} d(y', x'_i) = d(y', x_0) + a_d.
    \end{equation*}

    Taking the expectation, the asymptotic expected egalitarian cost of the mechanism is bounded below by:
   
    \begin{equation*}
        \lim_{n \to \infty} \SC(\mech(\vec{x}'), \vec{x}') = \lim_{n \to \infty} \E\left[\max_{i \in [n]} d(Y', x'_i)\right] = \E[d(Y', x_0)] + a_d \ge 1 + a_d.
    \end{equation*}

Since the optimal egalitarian cost for the modified profile $\vec{x}'$ is $\opt(\vec{x}') = a_d$, the asymptotic approximation ratio is at least: 
\[
        \frac{1 + a_d}{a_d} = 1 + \frac{1}{a_d} 
         = 1 + \frac{1}{\sqrt{2 + \frac{2}{d}}} 
        = 1 + \sqrt{\frac{d}{2(d+1)}}. 
\]
\end{proof}

\begin{remark}
    The established bound represents a natural limit for this specific deviation strategy. The proof requires an initial configuration of agents that are equidistant to its geometric median while maximizing the ratio of its circumradius to its diameter. By Jung's theorem, the maximum possible such ratio for any set in $\R^d$ is $\sqrt{\nicefrac{d}{2(d+1)}}$. As the regular simplex achieves this limit, any further improvement to the lower bound will require a different proof strategy.
\end{remark}

Additionally, for $d = 2$, we can parameterize our lower bound construction in terms of the number of agents in a cluster: Suppose in the initial profile $\vec{x}$ we have $n = 3k$ agents distributed equally over the three vertices of the simplex (which form an equilateral triangle), and we let the $k$ agents in $x_0$ deviate equidistantly to the surface of $C_{x_0}$. We can then lower bound the maximum distance of $y'$ to any of these points by considering the point $x'$ on $C_{x_0}$ that is nearest to the point that we obtain when projecting $y'$ through the center $x_0$ onto $C_{x_0}$.  

These lower bounds are summarized in \Cref{tab:lower_bounds_n} and formalized in \Cref{thm:discretized} below, whose proof can be found in the appendix. 

\begin{restatable}{theorem}{LBdiscretized}
\label{thm:discretized}
Any randomized strategyproof mechanism for facility location in $\R^2$ with $n \ge 6$ agents has an egalitarian approximation ratio of at least $\sqrt{\nicefrac{4}{3}+\nicefrac{2}{\sqrt{3}}\cos(\pi/(n/3))}$.
\end{restatable}

\begin{table}[htbp]
\centering
\begin{minipage}{0.45\textwidth}
\centering
\begin{tabular}{|c|c|}
    \hline
    $d$ & Lower bound \\
    \hline 
    $1$ & $1 + \sqrt{\nicefrac{1}{4}} = 1.5$ \\[1ex]
    \hline        
    $2$ & $1 + \sqrt{\nicefrac{2}{6}} \approx 1.577$ \\[1ex]
    \hline
    $3$ & $1 + \sqrt{\nicefrac{3}{8}} \approx 1.612$ \\[1ex]
    \hline
    $\infty$ & $1 + \sqrt{\nicefrac{1}{2}} \approx  1.707$ \\[1ex]
    \hline
\end{tabular}
\caption{Asymptotic lower bounds as a function of the dimension $d$.}
\label{tab:lower_bounds}
\end{minipage}\hfill
\begin{minipage}{0.45\textwidth}
\centering
\begin{tabular}{|c|c|}
    \hline
    $n$ & Lower bound ($d = 2$) \\
    \hline
    $6$  & $\sqrt{\nicefrac{4}{3}} \approx 1.155$ \\[1ex]
    \hline
    $9$  & $\sqrt{\nicefrac{4}{3} + \nicefrac{2}{\sqrt{3}}\cos\!\left(\nicefrac{\pi}{3}\right)} \approx 1.382$ \\[1ex]
    \hline
    $12$ & $\sqrt{\nicefrac{4}{3} + \nicefrac{2}{\sqrt{3}}\cos\!\left(\nicefrac{\pi}{4}\right)} \approx 1.466$ \\[1ex]
    \hline
    $15$ & $\sqrt{\nicefrac{4}{3} + \nicefrac{2}{\sqrt{3}}\cos\!\left(\nicefrac{\pi}{5}\right)} \approx 1.506$ \\[1ex]
    \hline
\end{tabular}
\caption{Lower bounds for $d=2$ as a function of the number of agents $n$.}
\label{tab:lower_bounds_n}
\end{minipage}
\end{table}


We can generalize this result to an arbitrary dimension $d \ge 1$. The main difficulty is that there are multiple ways to distribute the deviating agents over the sphere $C_{x_0}$. A (suboptimal) non-regular discretization is to use the \emph{spherical coordinate grid}: We discretize $C_{x_0}$ using $d-1$ angles $\theta_1, \dots, \theta_{d-1}$, where each $\theta_i \in [0, \pi]$ except $\theta_{d-1} \in [0,2\pi]$, and then place $(n/(d+1))^{1/(d-1)}$ equally spaced values on each angular coordinate, giving $k = n/(d+1)$ points. Using this discretization, we obtain the following lower bound, whose proof is deferred to the appendix.

\begin{restatable}{theorem}{LBddim}
\label{thm:d-dim-lowerbound}
For any dimension $d \ge 2$, the approximation ratio of any strategyproof mechanism with $n \ge (d+1)(\sqrt{d-1})^{d-1}$ agents is lower bounded by
\[
    \sqrt{\frac{d}{2(d+1)} + 2\sqrt{\frac{d}{2(d+1)}}\cos\!\left(\frac{\pi\sqrt{d-1}}{(n/(d+1))^{1/(d-1)}}\right) + 1}.
\]
\end{restatable}

\subsection{Results for Few Agents}
The lower bounds from the previous subsection demonstrate that the $1.5$ randomized approximation ratio achievable in the line is unattainable in higher dimensions, even for a small number of agents (for instance, when $n = 15$ in $\mathbb{R}^2$). However, a limitation of our results is that for small numbers of agents, we obtain no bound.

To address this, we prove a lower bound of $1.277$ for $n \ge 3$ agents in $\R^2$. Similar to Theorem \ref{lowerboundtheorem}, the argument relies on a base profile $\vec{x} = (x_1, x_2, x_3)$ located at the vertices of an equilateral triangle inscribed in the unit circle:
$$x_1 = (0, 1), \quad x_2 = \left(\frac{\sqrt{3}}{2}, -\frac{1}{2}\right), \quad x_3 = \left(-\frac{\sqrt{3}}{2}, -\frac{1}{2}\right)$$
and a deviation profile $\vec{x}' = (x_1', x_2, x_3)$, where agent 1 misreports to $x_1' = (0, 1+\sqrt{3})$. In this modified profile, all three agents are at a distance of $\sqrt{3}$ from the point $x_1$, corresponding to an optimal egalitarian cost of $\opt(\vec{x}') = \sqrt{3}$.

To evaluate the mechanism's performance under this deviation, the following lemma characterizes the minimum possible social cost on $\vec{x}'$ as a function of the facility's distance to $x_1$, which is the optimal location in this new profile. The proof of this lemma is deferred to the appendix.

\begin{restatable}{lemma}{MinCostLemma}
\label{lem:min_cost}
    Consider the profile $\vec{x}' = (x_1', x_2, x_3)$ with $x_1' = (0, 1+\sqrt{3})$, $x_2 = \left(\frac{\sqrt{3}}{2}, -\frac{1}{2}\right)$, and $x_3 = \left(-\frac{\sqrt{3}}{2}, -\frac{1}{2}\right)$. Let $L(r)$ denote the minimum possible egalitarian cost for $\vec{x}'$ subject to the constraint that the facility is located at distance $r$ from $x_1 = (0,1)$. For $r \ge 0$, this function is strictly convex and increasing, given by:
    $$L(r) = \sqrt{r^2 + \lambda r + 3}, \quad \text{where } \lambda = \frac{2\sqrt{3}}{\sqrt{8+4\sqrt{3}}}.$$
\end{restatable}

\begin{theorem} \label{thm:lower_bound}
    Any randomized strategyproof mechanism $\mech$ for facility location in $\mathbb{R}^2$ has an approximation ratio of at least $1.277$ for $n \ge 3$ agents.
\end{theorem}

\begin{proof}
    Let $\vec{x}$ and $\vec{x}'$ be as defined above. By Lemma~\ref{lem:expected_cost_simplex}, the expected distance from the facility to at least one of these vertices is at least $1$. Without loss of generality, assume $\E[d(\mech(\vec{x}), x_1)] \ge 1$.

    Let $R = d(\mech(\vec{x}'), x_1)$ be the random variable corresponding to the distance from $x_1$ to the facility location $\mech(\vec{x}')$. By strategyproofness, agent~1 cannot decrease their expected distance to their true location by misreporting:
    $$ \E[R] \ge \E[d(\mech(\vec{x}), x_1)] \ge 1. $$

    By Lemma~\ref{lem:min_cost}, for a fixed value of $R$, the maximum distance from any realized facility to the agents in $\vec{x}'$ is bounded below by $L(R)$. Applying Jensen's inequality to the convex function $L$, the expected maximum cost for the reported profile $\vec{x}'$ is bounded by:
    $$ \SC(\mech(\vec{x}'), \vec{x}') \ge \E[L(R)] \ge L(\E[R]) \ge L(1) = \sqrt{4 + \lambda}. $$

    Dividing this expected cost by the optimal cost for $\vec{x}'$ we obtain a lower bound of $ \frac{\sqrt{4 + \lambda}}{\sqrt{3}} \approx 1.277$ for the approximation ratio. This lower bound extends to $n > 3$ by co-locating any additional agents at the origin and applying the same argument.
\end{proof}

\begin{remark}
This approach could be generalized to higher dimensions or other values of $n$, but the optimization gets significantly more complex with diminishing gains in the bounds.
\end{remark}

We now focus on the two-agent case in $\mathbb{R}^d$ for $d \ge 2$, where we present a randomized mechanism that outperforms the general lower bound of Theorem~\ref{lowerboundtheorem}. The optimal facility always lies on the line connecting the two reported locations. Yet, if a mechanism restricts its output to this line, it cannot achieve an approximation ratio better than $1.5$ \cite{procaccia_approximate_2013}. By placing the facility outside this axis, our mechanism bypasses this restriction to achieve an improved ratio.

\begin{definition}[Orthogonal Sphere Mechanism]\label{def:orthogonal_sphere}
    Given reports $x_1, x_2$, let $m = \frac{x_1 + x_2}{2}$ and $h = \frac{\|x_1 - x_2\|}{2}$. The mechanism outputs $L = m + \vec{v}$, where $\vec{v}$ is drawn uniformly from a $(d-2)$-dimensional sphere $S^\perp$ of radius $h$ centered at the origin within the hyperplane orthogonal to the vector $x_1 - x_2$.
    
    (For $d=2$, $S^\perp$ consists of two vectors of length $h$, each chosen with probability $\nicefrac{1}{2}$.)
\end{definition}

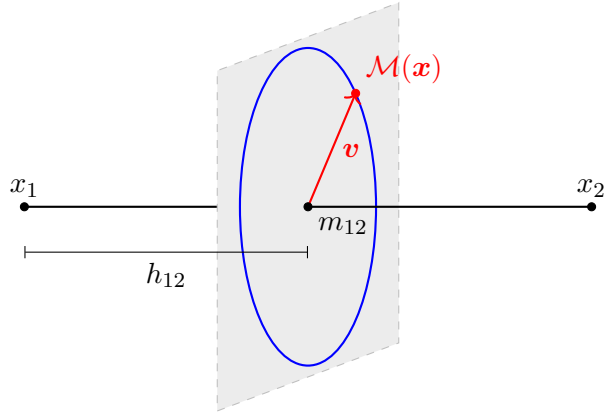
\begin{figure}[htbp]
    \centering
    \begin{tikzpicture}[scale=1.5]
        
        \draw[thick] (-2.5, 0) -- (0, 0); 
        
        \filldraw[fill=gray!15, draw=gray!50, dashed] (-0.8, -1.8) -- (0.8, -1.2) -- (0.8, 1.8) -- (-0.8, 1.2) -- cycle;
        
        \draw[thick, blue] (0,0) ellipse (0.6cm and 1.4cm);
        
        \draw[thick] (0, 0) -- (2.5, 0);
        
        \coordinate (X1) at (-2.5, 0);
        \coordinate (X2) at (2.5, 0);
        \coordinate (M) at (0, 0);
        \coordinate (L) at (0.42, 1.0); 
        
        \draw[thick, red, ->] (M) -- (L) node[midway, right] {$\vec{v}$};
        \filldraw[red] (L) circle (1pt) node[above right] {$\mech({\vec{x}})$};
        
        \filldraw (X1) circle (1pt) node[above] {$x_1$};
        \filldraw (X2) circle (1pt) node[above] {$x_2$};
        \filldraw (M) circle (1pt) node[below right] {$m_{12}$};
        
        \draw[|-|] (-2.5, -0.4) -- (0, -0.4) node[midway, below] {$h_{12}$};
        
    \end{tikzpicture}
    \caption{The 2-Agent Orthogonal Sphere Mechanism. The facility is chosen uniformly at random from the orthogonal sphere (blue).}
    \label{fig:ortho_sphere_3d}
\end{figure}

\begin{theorem}
\label{thm:2agent:approx}
    The Orthogonal Sphere Mechanism is strategyproof in expectation and has an approximation ratio of $\sqrt2$.
\end{theorem}

\begin{proof}
    We establish the two claims of the theorem separately.

    \textbf{Approximation Ratio.} For any profile {$\vec{x} = (x_1, x_2)$}, the mechanism outputs {$L = m + \vec{v}$}, where {$m = \frac{x_1 + x_2}{2}$} and {$\|\vec{v}\| = h$}. For any vector $\vec{v}$ sampled from the orthogonal sphere $S^\perp$, we have {$\langle m - x_1, \vec{v} \rangle = 0$}. By the Pythagorean theorem, the distance from the facility to Agent~1 is:
    $$\|L - x_1\| = \sqrt{\|m - x_1\|^2 + \|\vec{v}\|^2} = \sqrt{h^2 + h^2} = h\sqrt{2}.$$
    By symmetry, the distance to Agent~2 is also {$\|L - x_2\| = h\sqrt{2}$}. Since the optimal egalitarian cost is $h$, the approximation ratio evaluates to exactly $\frac{h\sqrt{2}}{h} = \sqrt{2}$.

    \textbf{Strategyproofness.} Without loss of generality, by translating and rotating the coordinate system, we can set {$x_1 = (-R, 0, \dots, 0)$} and {$x_2 = (R, 0, \dots, 0)$} for {$R \ge 0$}. This gives an expected distance of {$h\sqrt{2} = R\sqrt{2}$} to the facility for both agents. To establish strategyproofness, it suffices to show that any deviation by Agent~1 results in an expected distance of at least $\sqrt{2}R$ from the facility to him.
    
    Because the mechanism's configuration is rotationally symmetric around the line between $x_1$ and $x_2$, any misreported location can be rotated into the $xy$-plane without loss of generality. Thus, we assume Agent~1 misreports to a location {$$x_1' = (-R + \delta_x, \delta_y, 0, \dots, 0).$$} The mechanism computes the reported midpoint {$m = (\frac{\delta_x}{2}, \frac{\delta_y}{2}, 0, \dots, 0)$} and a reported radius $h$ satisfying {$h^2 = (R - \frac{\delta_x}{2})^2 + (\frac{\delta_y}{2})^2$}. It then samples a vector {$\vec{v} = (v_x, v_y, \dots, v_d) \in S^\perp$}. Since $\vec{v}$ must be orthogonal to the vector {$x_2 - x_1' = (2R - \delta_x, -\delta_y, 0, \dots, 0)$}, its components must satisfy:
    \begin{equation}
        \label{eq:orthogonality_identity}
        v_y \delta_y = (2R - \delta_x)v_x.
    \end{equation}

    Pairing every sampled vector $\vec{v}$ with its antipodal counterpart $-\vec{v}$, the corresponding squared distances from the true location $x_1$ are:
    $$D_{\pm} = \|(m \pm \vec{v}) - x_1\|^2 = \|m - x_1\|^2 + \|\vec{v}\|^2 \pm 2\langle m - x_1, \vec{v} \rangle.$$
    Using {$m - x_1 = (R + \frac{\delta_x}{2}, \frac{\delta_y}{2}, 0, \dots, 0)$} and {$\|\vec{v}\|^2 = h^2$}, expanding the norm gives:
    \begin{equation*}
        \|m - x_1\|^2 + h^2 = 2R^2 + \frac{\delta_x^2}{2} + \frac{\delta_y^2}{2}.
    \end{equation*}
    For the inner product, substituting the identity \eqref{eq:orthogonality_identity} into the expression gives:
    \begin{equation*}
        \langle m - x_1, \vec{v} \rangle = v_x\left(R + \frac{\delta_x}{2}\right) + v_y\left(\frac{\delta_y}{2}\right) = v_x\left(R + \frac{\delta_x}{2}\right) + v_x\left(R - \frac{\delta_x}{2}\right) = 2Rv_x.
    \end{equation*}

    Combining these terms gives {$D_{\pm} = 2R^2 + \frac{\delta_x^2}{2} + \frac{\delta_y^2}{2} \pm 4R v_x$}. This can be rewritten as:
    $$D_+ = 2\left[ (R + v_x)^2 + z^2 \right] \quad \text{and} \quad D_- = 2\left[ (R - v_x)^2 + z^2 \right],$$
    where {$z = \sqrt{\frac{\delta_x^2 + \delta_y^2}{4} - v_x^2}$}. To guarantee that $z$ is a well-defined real number, note that {$v_x^2 + v_y^2 \le h^2$}. Isolating $v_y$ via \eqref{eq:orthogonality_identity} results in {$v_x^2 [ 1 + (2R - \delta_x)^2/\delta_y^2 ] \le h^2$}, which simplifies using the definition {$4h^2 = (2R - \delta_x)^2 + \delta_y^2$} to show that:
    \begin{equation*}
        v_x^2 \le \frac{\delta_y^2}{4} \le \frac{\delta_x^2 + \delta_y^2}{4}.
    \end{equation*}

   These expressions have a geometric interpretation in $\mathbb{R}^2$. Consider the auxiliary points {$p = (v_x, z)$}, {$F_1 = (-R, 0)$}, and {$F_2 = (R, 0)$}. Noting that {$\|p - F_1\|^2 = (R + v_x)^2 + z^2$} and {$\|p - F_2\|^2 = (R - v_x)^2 + z^2$}, we can rewrite the paired distances as:
\begin{equation*}
    \sqrt{D_+} = \sqrt{2}\|p - F_1\| \quad \text{and} \quad \sqrt{D_-} = \sqrt{2}\|p - F_2\|.
\end{equation*}

Applying the triangle inequality to these auxiliary points, the average distance of the antipodal pair {$\{\vec{v}, -\vec{v}\}$} to Agent~1's true location is bounded as follows:
\begin{align*}
    \frac{d(m + \vec{v}, x_1) + d(m - \vec{v}, x_1)}{2} &= \frac{\sqrt{D_+} + \sqrt{D_-}}{2} \\
    &= \frac{\sqrt{2}}{2} \left( \|p - F_1\| + \|p - F_2\| \right) \\
    &\ge \frac{\sqrt{2}}{2} \|F_1 - F_2\| \\
    &= \sqrt{2}R.
\end{align*}
    Because this lower bound holds pointwise for every antipodal pair on the sphere, taking the expectation over the uniform distribution of $S^\perp$ preserves the inequality:
    $$\E[d(L, x_1)] = \E\left[\frac{\sqrt{D_+} + \sqrt{D_-}}{2}\right] \ge \sqrt{2}R,$$
    completing the proof.
\end{proof}

\begin{remark}
    \Cref{thm:2agent:approx} proves a separation between randomized and deterministic mechanisms for the two-agent case in $\R^2$: the randomized $\sqrt{2}$-approximation beats the tight lower bound of $2$ for deterministic mechanisms~\cite{ChanLW26,GoelH23}.
    Furthermore, the theorem also shows that the two-agent case is \enquote{easier} in $\R^2$ than it is on the line. In the latter, the best-possible approximation ratio is $1.5 > \sqrt{2}$~\cite{procaccia_approximate_2013}.
\end{remark}

\begin{remark}
    We remark that the natural generalization of the Orthogonal Sphere Mechanism to $n > 2$, which picks two agents uniformly at random and then executes the Orthogonal Sphere Mechanism on these two agents, does not achieve an approximation ratio better than $2$. To see this, consider $n-1$ reports at $(0,0)$ and one report at $(0,1)$. For $n$ towards $\infty$, the approximation ratio of the generalized mechanism approaches $2$ on this family of instances.
\end{remark}

\section{Dimension Augmented}

We now focus on the output augmented setting where $\mathcal{I} \subsetneq \mathcal{O}$. We study the cases where the agents' locations are restricted to a line or a circle, but the output space is expanded to the Euclidean plane ($\mathcal{O}=\mathbb{R}^2$). As we will see, giving the mechanism this extra room to maneuver allows it to achieve better approximation ratios than those possible in the classical setting.

\subsection{Line Case}

The agents' true locations are all in the $y=0$ line, i.e., $\mathcal{I} = \R \times \{0\}$, and the facility can be placed anywhere in $\R^2$, i.e., $\mathcal{O} = \R^2$.

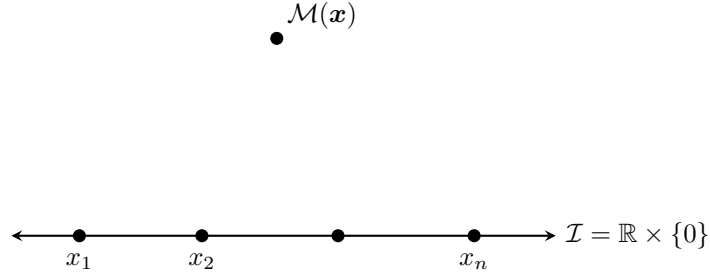
\begin{figure}[htbp]
    \centering
    \begin{tikzpicture}[scale=0.9]

        \draw[thick, <->, >=stealth] (-4, 0) -- (4, 0) node[right, font=\small] {$\mathcal{I} = \mathbb{R} \times \{0\}$};

        \filldraw (-3, 0) circle (2.5pt) node[below=3pt, font=\small] {$x_1$};
        \filldraw (-1.2, 0) circle (2.5pt) node[below=3pt, font=\small] {$x_2$};
        \filldraw (0.8, 0) circle (2.5pt) node[below=3pt, font=\small] {};
        \filldraw (2.8, 0) circle (2.5pt) node[below=3pt, font=\small] {$x_n$};

        \filldraw[black] (-0.1, 2.9) circle (2.5pt) node[above right, font=\small, color=black] {$\mech(\vec{x})$};

    \end{tikzpicture}
    \caption{Illustration of dimension-augmented line setting.}
    \label{fig:framework_dimension}
\end{figure}

\subsubsection{Deterministic Mechanisms}

Inspired by the 2-agent Orthogonal Sphere Mechanism (Definition \ref{def:orthogonal_sphere}) from the previous chapter, we introduce the Augmented Midpoint Mechanism. By making use of the additional dimension, this mechanism achieves an approximation ratio of $\sqrt{2}$ for any number of agents, which we prove is the best a deterministic mechanism can achieve.

\begin{definition}[Augmented Midpoint Mechanism]
    Let $\vec{x} = (x_1, \dots, x_n) \in \mathbb{R}^n$ be the profile of reports on the x-axis. Let $l = \min_i x_i$ and $r = \max_i x_i$. The Augmented Midpoint Mechanism outputs the facility location $Y = \left(\frac{l+r}{2}, \frac{r-l}{2}\right) \in \mathbb{R}^2$.
\end{definition}

The strategyproofness of the mechanism relies on balancing the horizontal and vertical distance from an extreme agent to the facility, which is illustrated in Figure \ref{fig:augmented_midpoint}. If an agent misreports outward to pull the horizontal midpoint closer, it forces the facility higher up, which offsets the horizontal gain. Conversely, reporting inward to pull the facility lower towards the line shifts the midpoint away, increasing the horizontal distance and canceling out the benefit of the lower height. We formalize this in the following theorem.

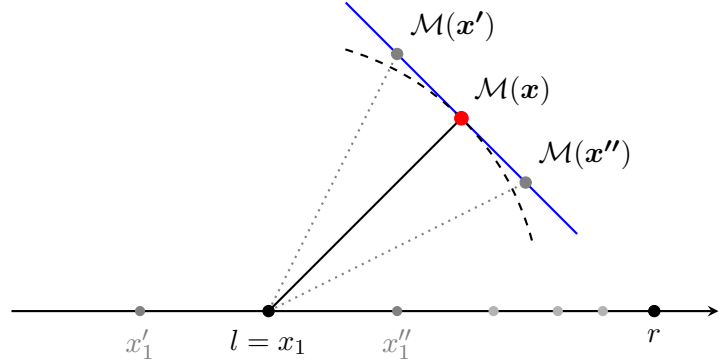
\begin{figure}[htbp]
    \centering
    \begin{tikzpicture}[scale=0.85]
        \draw[thick, ->, >=stealth] (-8, 0) -- (3, 0);
        
        \coordinate (L) at (-4, 0);      
        \coordinate (R) at (2, 0);       
        \coordinate (Opt) at (-1, 3);    
        
        \coordinate (L_out) at (-6, 0);  
        \coordinate (Opt_out) at (-2, 4);
        
        \coordinate (L_in) at (-2, 0);   
        \coordinate (Opt_in) at (0, 2);
        
        \draw[thick, blue] (-2.8, 4.8) -- (0.8, 1.2) node[below right, text=gray] {};
        
        \draw[thick, black] (L) -- (Opt);
        \draw[thick, gray, dotted] (L) -- (Opt_out);
        \draw[thick, gray, dotted] (L) -- (Opt_in);
        
        \draw[black, thick, dashed] (-4,0) ++(15:4.242) arc[start angle=15, end angle=75, radius=4.242];

        \filldraw[black] (L) circle (2.5pt) node[below=3pt] {$l = x_1$};
        \filldraw[black] (R) circle (2.5pt) node[below=3pt] {$r$};
        \filldraw[gray] (L_out) circle (2pt) node[below=3pt] {$x_1'$};
        \filldraw[gray] (L_in) circle (2pt) node[below=3pt] {$x_1''$};
        
        \filldraw[red] (Opt) circle (3pt) node[above right=1pt, text=black] {$\mech(\vec{x})$};
        \filldraw[gray] (Opt_out) circle (2.5pt) node[above right=1pt, text=black] {$\mech(\vec{x'})$};
        \filldraw[gray] (Opt_in) circle (2.5pt) node[above right=1pt, text=black] {$\mech(\vec{x''})$};

        \filldraw[gray!60] (-0.5, 0) circle (2pt);
        \filldraw[gray!60] (0.5, 0) circle (2pt);
        \filldraw[gray!60] (1.2, 0) circle (2pt);
        
    \end{tikzpicture}
    \caption{The Augmented Midpoint Mechanism under unilateral deviations by Agent~1. Misreporting shifts the output along the blue line. The dashed arc indicates all points that are at the same distance from the truthful report $x_1$ as the output $\mech(\vec{x})$. Any outward ($x_1'$) or inward ($x_1''$) deviation pushes the output outside this arc, increasing the agent's cost.}
    \label{fig:augmented_midpoint}
\end{figure}

\begin{restatable}{theorem}{thmAugmentedLine}
    \label{thm:augmented:line}
    The Augmented Midpoint Mechanism is strategyproof and has an approximation ratio of $\sqrt2$.
\end{restatable}

\begin{proof}
    Let $\mech$ denote the mechanism. By translation and scale invariance of $\mech$, we can assume, without loss of generality, the extreme reports are $l = -1$ and $r = 1$. We establish the two properties of the mechanism separately:
    \begin{itemize}
        \item \textbf{Approximation Ratio:} The optimal facility location is the origin $(0,0)$, yielding a maximum cost of $1$. The mechanism $\mech$ outputs $(0,1)$. For any agent $i$ with $x_i \in [-1, 1]$, the distance to the output is $\sqrt{x_i^2 + 1}$, which is maximized at the extremes $x_i = \pm 1$ with a value of $\sqrt{2}$. Thus, the approximation ratio is $\sqrt{2}$.

        \item 
        \textbf{Strategyproofness:} By the strategyproofness of the 2-agent mechanism (\Cref{thm:2agent:approx}), the extreme agents $l$ and $r$ cannot profit by deviating. Therefore, we only need to consider an interior agent misreporting to become a new extreme.

        Suppose an interior agent is located at $(a, 0)$ with $a \in [0, 1)$. If the agent reports a new extreme $z < -1$, the facility shifts to the location $\left(\frac{z+1}{2}, \frac{1-z}{2}\right)$. The distance from the agent to this new location is $a - \frac{z+1}{2} > a$ along the $x$-axis and $\frac{1-z}{2} > 1$ along the $y$-axis. Since both coordinate-wise distances are larger than the distances to the truthful output facility at $(0,1)$, this deviation is unprofitable.
        The agent can also define a new extreme by reporting $z > 1$. This moves the facility to the coordinates $\left(\frac{z-1}{2}, \frac{z+1}{2}\right)$. For this deviation to be profitable, the squared distance to this new position must be less than the squared distance to the truthful output which gives us the inequality:
$$ \left(\frac{z-1}{2} - a\right)^2 + \left(\frac{z+1}{2}\right)^2 < a^2 + 1. $$
Simplifying reduces this inequality to $z^2 - 2az + 2a - 1 < 0$, which only holds when $z \in (2a - 1, 1)$. This contradicts our requirement that $z > 1$. A symmetric argument holds for an interior agent located at $(a,0)$ with $a<0$. Thus, $\mech$ is strategyproof.
\end{itemize}
    
\end{proof}

\begin{remark}
The mechanism achieves a better approximation ratio than what is possible under strategyproofness in the standard line setting. Remarkably, this improvement occurs even though the optimal location always lies on the line at $\left(\frac{l+r}{2}, 0\right)$, while the mechanism always outputs a location outside of it (unless the reported profile is unanimous).
\end{remark}

We complement our mechanism with a matching lower bound. More specifically, we prove a lower bound of $\sqrt2$ for the approximation ratio of any deterministic mechanism for the $n=2$ case which can be extended to an arbitrary number of agents.

\begin{theorem}\label{thm:main_bound}
    Let $\mech$ be a deterministic, strategyproof mechanism for $n \ge 2$ agents in the dimension-augmented line setting that achieves an approximation ratio of $\lambda$. Then, $\lambda \ge \sqrt{2}$.
\end{theorem}

\begin{proof}
    It suffices to prove the bound for $n=2$, as it extends to any $n \ge 2$ via \Cref{lem:population_extension2}. For the proof, fix $x_1=0$ and let $x_2 = x \in (0, \infty)$. The mechanism's output defines a curve parameterized by $x$:
    $$f(x) := \mech(0, x) = (u(x), v(x)).$$
    
  The optimal facility location for this profile is the midpoint $(x/2, 0)$, with a corresponding maximum cost of $x/2$. Because $\mech$ guarantees a $\lambda$-approximation, the squared distance from $\mech(\vec{x})$ to either agent is bounded by $(\lambda x/2)^2$. For Agents 1 and 2 respectively, this implies:
\begin{align}
    d(x_1, \mech(\vec{x})) = u(x)^2 + v(x)^2 &\le \frac{\lambda^2 x^2}{4}, \label{eq:bound_u} \\
    D(x) := d(x_2, \mech(\vec{x})) = (x - u(x))^2 + v(x)^2 &\le \frac{\lambda^2 x^2}{4},\label{eq:bound_D}
\end{align}
    where we defined $D(x)$ as the squared distance from the point $x_2 = (x,0)$ to the output of the mechanism $\mech(x_1, x)$.

    By strategyproofness, Agent 2 located at $x$ cannot improve their outcome by misreporting $y \in (0, \infty)$, meaning $D(x) \le (x - u(y))^2 + v(y)^2$. To express this in terms of $D(y)$, we expand the squared term and substitute $v(y)^2 = D(y) - (y - u(y))^2$:
    \begin{align*}
        D(x) &\le x^2 - 2x u(y) + u(y)^2 + v(y)^2 \\
             &= x^2 - 2x u(y) + D(y) - y^2 + 2y u(y) \\
             &= D(y) + (x^2 - y^2) - 2u(y)(x - y).
    \end{align*}
    Using the identity $x^2 - y^2 = (x - y)^2 + 2y(x - y)$, we obtain the upper bound:
    \begin{equation}\label{eq:sp_upper}
        D(x) - D(y) \le (x - y)^2 + 2(x - y)(y - u(y)).
    \end{equation}

    Symmetrically, preventing an agent at $y$ from misreporting $x$ requires $D(y) \le (y - u(x))^2 + v(x)^2$. Expanding this similarly yields:
    $$D(y) \le D(x) + (y^2 - x^2) - 2u(x)(y - x).$$
    In a similar way, by rearranging the inequality we obtain:
    \begin{equation}\label{eq:sp_lower}
        D(x) - D(y) \ge -(x - y)^2 + 2(x - y)(x - u(x)).
    \end{equation}
    
    Let $[a, b] \subset (0, \infty)$ be an arbitrary compact interval. For any $z \in [a, b]$, inequality \eqref{eq:bound_u} ensures $|u(z)| \le \frac{\lambda}{2}z \le \frac{\lambda}{2}b$. Dividing \eqref{eq:sp_upper} and \eqref{eq:sp_lower} by $|x - y|$ we get that for any distinct $x, y \in [a, b]$, the absolute difference quotient is bounded:
    $$ \left| \frac{D(x) - D(y)}{x - y} \right| \le |x - y| + 2|x| + 2|y| + 2|u(x)| + 2|u(y)| \le (b - a) + 4b + 2\lambda b. $$
    Because for any fixed interval $[a,b] \subseteq (0, \infty)$ this upper bound is a constant, $D(x)$ is locally Lipschitz on $(0, \infty)$. This implies that for any compact interval $[a,b] \subset (0, \infty)$, $D(x)$ is absolutely continuous and therefore differentiable almost everywhere (a.e.).
    
    Let $x$ be a point of differentiability. Dividing \eqref{eq:sp_lower} by $(x - y)$ and taking the limits as $y \to x^+$ and $y \to x^-$ we obtain:
 \begin{equation}
    \label{derivative2}
    D'(x) = 2(x - u(x)) \quad \text{a.e.}
 \end{equation}

    Expanding $D(x)$ and applying \eqref{eq:bound_u} we obtain:

    \begin{equation}
    \label{stepineq}
    D(x) = x^2 - 2x u(x) + u(x)^2 + v(x)^2 \le x^2 - 2x u(x) + \frac{\lambda^2 x^2}{4}.
    \end{equation}
    By the derivative equation \eqref{derivative2}, we substitute $-2x u(x) = -x(2x - D'(x)) = -2x^2 + x D'(x)$ into inequality \eqref{stepineq} and obtain:
    $$D(x) \le x^2 - 2x^2 + x D'(x) + \frac{\lambda^2 x^2}{4} = -x^2 + x D'(x) + \frac{\lambda^2 x^2}{4} \quad \text{a.e.}$$
    Rearranging terms to group $D(x)$ and its derivative gives:
    $$x D'(x) - D(x) \ge x^2\left(1 - \frac{\lambda^2}{4}\right) \quad \text{a.e.}$$
    Dividing both sides by $x^2 >0$ we obtain the derivative of a quotient on the left hand side:
    $$\frac{d}{dx} \left( \frac{D(x)}{x} \right) \ge 1 - \frac{\lambda^2}{4} \quad \text{a.e.}$$

    Because $D(x)$ is absolutely continuous on any compact subset of $(0, \infty)$, the quotient $D(x)/x$ is absolutely continuous on any interval $[\epsilon, X] \subset (0, \infty)$. Integrating from $\epsilon$ to $X$ preserves the above inequality and we obtain:
    \begin{equation}
    \label{integrated}
    \frac{D(X)}{X} - \frac{D(\epsilon)}{\epsilon} \ge \left( 1 - \frac{\lambda^2}{4} \right)(X - \epsilon).
    \end{equation}
    
    From \eqref{eq:bound_D}, we know $0 \le D(\epsilon) \le \lambda^2 \epsilon^2 / 4$, which implies $0 \le \frac{D(\epsilon)}{\epsilon} \le \frac{\lambda^2}{4} \epsilon$. Thus, $\lim_{\epsilon \to 0^+} \frac{D(\epsilon)}{\epsilon} = 0$. Taking this limit, multiplying the resulting inequality \eqref{integrated} by $X$, and combining it with the original bound \eqref{eq:bound_D} we obtain:
    $$X^2\left(1 - \frac{\lambda^2}{4}\right) \le D(X) \le \frac{\lambda^2 X^2}{4}.$$
    Dividing by $X^2$ simplifies to:
    $$1 - \frac{\lambda^2}{4} \le \frac{\lambda^2}{4},$$
    which establishes $\lambda \ge \sqrt{2}$.

    It remains to argue that the bound extends to any $n > 2$.
 If an $n$-agent mechanism achieved an approximation ratio better than $\lambda$, we could construct a strategyproof 2-agent mechanism preserving this improved bound by having it run the $n$-agent mechanism on a profile where the reports $(a,b)$ are padded with $n-2$ dummy agents co-located at $b$. The full formal argument for the extension to $n>2$ can be found in~\Cref{lem:population_extension2} in the appendix. 
 
\end{proof}

\paragraph{Generalization to $\mathcal{I} = \mathbb{R}^d \times \{ 0\}$ and $\mathcal{O}=\mathbb{R}^{d+1}$.} We have seen that the Augmented Midpoint Mechanism achieves the best-possible approximation ratio of $\sqrt{2}$ among deterministic mechanisms when $\mathcal{I} = \mathbb{R} \times \{0\}$ and $\mathcal{O}=\mathbb{R}^{2}$. Next, we briefly discuss a natural generalization of the mechanism to the setting with $\mathcal{I} = \mathbb{R}^d \times \{ 0\}$ and $\mathcal{O}=\mathbb{R}^{d+1}$. To this end, consider the following mechanism:
\begin{enumerate}
    \item For each agent $i$, let $x_i = (x_1^i,\ldots, x^i_d,0)$ denote the report of $i$.
    \item For each coordinate $j \in [d]$, let $a_j = \min_{i \in [n]} x_j^i$ and $b_j = \max_{i \in [n]} x_j^i$. Let $m_j = \frac{a_j+b_j}{2}$ be the \emph{midpoint} of coordinate $j$ and let $r_j = \frac{b_j-a_j}{2}$ denote the \emph{radius} of coordinate $j$.
    \item Return $\mech(\vec{x}) = (m_1,\ldots,m_d,h)$ with $h = \sqrt{\sum_{j \in [d]} r_j^2}$.
\end{enumerate}

\begin{theorem}
\label{thm:augmented:d}
    The generalized mechanism is strategyproof and has an approximation ratio of $\sqrt{d+1}$ for the output augmented setting with $\mathcal{I} = \mathbb{R}^d \times \{ 0\}$ and $\mathcal{O}=\mathbb{R}^{d+1}$.
\end{theorem}

Note that this theorem subsumes~\Cref{thm:augmented:line} with $d=1$. For $d=2$, the mechanism achieves an approximation ratio of $\sqrt{3} \approx 1.732$, which beats the best-possible approximation factor of $2$~for deterministic mechanisms~\cite{ChanLW26,GoelH23} in two dimensions, i.e., $\mathcal{I} = \mathcal{O} = \mathbb{R}^2$. Starting from $d=3$, the mechanism does not improve upon deterministic mechanism with $\mathcal{I} = \mathcal{O} = \mathbb{R}^d$ anymore. 

\begin{lemma}
    The generalized mechanism is strategyproof.
\end{lemma}

\begin{proof}
    By definition, the mechanism outputs the facility at $\mech(\vec{x}) = (m_1,\ldots,m_d,h)$ with $h = \sqrt{\sum_{j \in [d]} r_j^2}$.
    The squared distance $D_i^2$ from agent $i$'s true location $(x^i_1, \ldots, x^i_d, 0)$ to the facility is:
    $$ D_i^2 = \sum_{j \in [d]} (m_j-x_j^i)^2 + \sum_{j \in [d]} r_j^2.$$
    We can decouple this into independent terms for each axis:
    $$ D_i^2 =  \sum_{j \in [d]} \left((m_j-x_j^i)^2 + r_j^2\right).$$
    Define $D_i^2(j) =  \left((m_j-x_j^i)^2 + r_j^2\right)$. The term $D_i^2(j)$ is the squared cost that agent $i$ would incur if we run the Augmented Midpoint Mechanism with $\mathcal{I}= \mathbb{R}$ and $\mathcal{O}= \mathbb{R}^2$ only for the $j$-coordinates.
    
    Because the mechanism computes $m_j, r_j$ using only the $j$-coordinates, each misreported coordinate of agent $i$ independently only affects the corresponding $D_i^2(j)$ term. By Theorem \ref{thm:augmented:line}, the Augmented Midpoint Mechanism is strategyproof, meaning a unilateral deviation cannot decrease the squared cost $D_i^2(j)$ for each coordinate $j$. That is $D_i'^2(j) \ge D_i^2(j)$ where $D_i'^2(j)$ is the squared cost of agent $i$ in coordinate $j$ after deviating.
    This immediately implies $D_i^2 = \sum_{j \in [d]} D_i^2(j) \le \sum_{j \in [d]} D_i'^2(j) = D_i'^2$. 
    Thus, the mechanism is strategyproof.    
\end{proof}

\begin{lemma}
    The generalized mechanism has an approximation ratio of $\sqrt{d+1}$.
\end{lemma}

\begin{proof}
     Note that the optimal facility location $z \in \mathbb{R}^{d+1}$ is of form $z = (z_1,\ldots,z_d,0)$ as all reported locations $x_i$ are of form $(x_1^i,\ldots x_d^i,0)$. Without loss of generality, let the optimal facility location be the origin $(0,\ldots,0)$. Let $R = \opt(\vec{x})$ denote the optimal radius. For every agent $i$, we have $\sum_{j \in [d]} (x_j^i)^2 \le R^2$.

    Because all agents lie within the optimal ball of radius $R$, the extreme reports on the axis $j$ must satisfy $r_j + |m_j| \le R$. Squaring both sides and rearranging yields
    $$ m_j^2 + r_j^2 \le R^2 - 2r_j|m_j|.$$

    The squared distance $D_i^2$ from agent $i$ to the facility $\mech(\vec{x}) = (m_1,\ldots,m_d,h)$ is
    \begin{align*}
         D_i^2 &= \sum_{j \in [d]} (m_j-x_j^i)^2 + \sum_{j \in [d]} r_j^2\\
         &= \sum_{j \in [d]} (x^i_j)^2 + \sum_{j \in [d]} (m_j^2 + r_j^2) - 2 \cdot \sum_{j \in [d]} x_j^i m_j.
    \end{align*}
    Plugging in $\sum_{j \in [d]} (x_j^i)^2 \le R^2$ and $ m_j^2 + r_j^2 \le R^2 - 2r_j|m_j|$ yields
    \begin{equation}
    \label{eq:augmented:d}
            D_i^2 \le  (d+1)R^2 - 2\sum_{j \in [d]} (x_j^i m_j + r_j|m_j|).
    \end{equation}    
    Next, we argue that $(x_j^i m_j + r_j|m_j|) \ge 0$. For each agent $i$ and coordinate $j$, we have $|x_j^i - m_j| \le r_j$ by choice of $r_j$ and $m_j$. Multiplying by $|m_j|$ gives 
    $$r_j|m_j| \ge |x_j^i - m_j||m_j|
           \ge -(x_j^i - m_j)m_j
           = m_j^2 - x_j^i m_j$$
    Rearranging this yields $x_j^i m_j + r_j|m_j| \ge m_j^2 \ge 0$.     

    Since  $(x_j^i m_j + r_j|m_j|) \ge 0$ for all agents $i$ and coordinates $j$, the inequality~\eqref{eq:augmented:d} implies
    $D_i^2 \le (d+1)R^2$ and, thus, $D_i \le \sqrt{d+1} \cdot R$. We can conclude that the mechanism is a $\sqrt{d+1}$-approximation.

    We finish the proof by arguing that this bound is tight. Consider an instance with $2d$ agents where, for each dimension $j \in [d]$, two agents are placed on the $j$-th coordinate axis: one at $R$ and one at $-R$. 


   The optimal solution places the facility at the origin with a cost of $\opt(\vec{x}) = R$. The mechanism, on the other hand, places the facility at $\mech(\vec{x}) = (m_1,\ldots,m_d,h) = (0,\ldots,0,h)$.

    Fix an arbitrary agent $i$ with report $x_i = (x_1^i, \dots, x_d^i,0)$. By construction, $x_i$ has exactly one non-zero coordinate $\pm R$. Since $m_j = 0$ and $r_j = R$ for all $j \in [d]$, we have:
    \[
    d(\mech(\vec{x}),x_i)
    = \sqrt{\sum_{j \in [d]} (x_j^i-m_j)^2 + h^2}
    = \sqrt{R^2 + d \cdot R^2}
    = \sqrt{d+1} \cdot R.
    \]
\end{proof}

\subsubsection{Lower Bound for Translation and Scale Invariant Mechanisms}





While our deterministic lower bound of $\sqrt{2}$ is tight, a natural next question is whether randomization can help achieve a better approximation ratio. Interestingly, randomized mechanisms face the same barrier when required to be scale and translation invariant. In this section, we establish this lower bound for the two-agent case, which extends to an arbitrary number of agents via~\Cref{lem:population_extension2}.

To formalize this analysis, we introduce the following notation for handling expectations over distributions with point masses. For a probability distribution $\mathcal{D}$, we denote the discrete probability mass at a single point $a$ as $\mathbb{P}_{\mathcal{D}}(a)$. To denote that an expectation ignores the individual contribution of a specific point $a$, we write the distribution shorthand as $\mathcal{D} \setminus \{a\}$. Formally, the resulting unnormalized expectation of a function $g$ is defined using the indicator function $\mathbf{1}_{y \neq a}$:
\begin{equation*}
\E_{y \sim \mathcal{D} \setminus \{a\}} [g(y)] = \E_{y \sim \mathcal{D}} [g(y) \cdot \mathbf{1}_{y \neq a}].
\end{equation*}

By translation and scale invariance, the output of a randomized mechanism $\mech$ on any two agent profile is completely determined by its output on a single reference configuration. More precisely, we can describe the output distribution of $\mech(x_1, x_2)$ on any reported profile $\vec{x}=(x_1,x_2)$ in terms of its output on the profile $(-1, 1)$.

For the remainder of this section, let $\mathcal{D} = \mech(-1,1)$ denote the mechanism's output distribution on the reference profile $(-1,1)$. For any reported profile $(l,r)$ with $l \le r$, scale and translation invariance imply that the output distribution is determined by the affine coordinate transformation $T_{l,r}: \mathbb{R}^2 \to \mathbb{R}^2$, defined as:
\[
T_{l,r}(a,b) = \left( \frac{r-l}{2}a + \frac{l+r}{2}, \; \frac{r-l}{2}b \right).
\]
That is, if the mechanism outputs a location $(a,b)$ with probability $p$ under the reference profile $(-1,1)$, it outputs $T_{l,r}(a,b)$ with probability $p$ under the profile $(l, r)$.

To establish our lower bound, we focus on the profile $\vec{x}$ where the agents are truthfully located at $x_1 = -1$ and $x_2 = 1$. If Agent 1 misreports their location as $x_1' \le 1$, the profile becomes $(x_1', 1)$ and a reference point $(a,b) \sim \mathcal{D}$ is mapped to the facility location $T_{x_1', 1}(a,b)$ given by:

$$T_{x_1', 1}(a,b) = \left( \frac{1-x_1'}{2}a + \frac{x_1'+1}{2}, \; \frac{1-x_1'}{2}b \right).
$$
Note that the distance from $x_1$ to this facility is given by $d(x_1, T_{x_1', 1}(a,b))$. Since the coordinates of $T_{x_1', 1}(a,b)$ are affine functions of $x_1'$, and the Euclidean norm is a convex function, their composition $d(x_1, T_{x_1', 1}(a,b))$ is convex with respect to $x_1'$. Consequently, Agent 1's expected distance to $x_1$ when he misreports $x_1'$,
$$
\E_{(a,b)\sim\mathcal{D}}[d(x_1, T_{x_1', 1}(a,b))],
$$
is also a convex function with respect to $x_1'$.

Furthermore, since we are only dealing with two agents, we assume without loss of generality that the distribution $\mathcal{D}$ is symmetric with respect to the $y$-axis. This follows because of the following lemma.

\begin{lemma} \label{lem:symmetric_reduction}
Let $\mech$ be a strategyproof, translation and scale invariant randomized mechanism for two agents, with output distribution $\mathcal{D}$ on the profile $(-1, 1)$ and approximation ratio $\lambda$. Then, there exists a strategyproof, translation and scale invariant randomized mechanism $\mech'$ whose corresponding output distribution $\mathcal{D}'$ on $(-1, 1)$ is symmetric with respect to the $y$-axis and also achieves an approximation ratio $\lambda' = \lambda$.
\end{lemma}

\begin{proof}
Let $\mathcal{D}_{\text{flipped}}$ be the distribution corresponding to the reflection of $\mathcal{D}$ across the $y$-axis, i.e. a location $(a,b) \sim \mathcal{D}$ is mapped to a location $(-a,b) \sim \mathcal{D}_{\text{flipped}}$ with the same probability. We define the new symmetric distribution $\mathcal{D}'$ as:
$$\mathcal{D}' = \frac{1}{2}\mathcal{D} + \frac{1}{2}\mathcal{D}_{\text{flipped}}.$$

Let $\mech'$ be the mechanism corresponding to this distribution $\mathcal{D}'$. Note that by symmetry, if $\mech$ is strategyproof then $\mech_{\text{flipped}}$ is also strategyproof. Furthermore, $\mech'(\vec{x}) = \frac{1}{2}\mech(\vec{x}) + \frac{1}{2}\mech_{\text{flipped}}(\vec{x})$. Thus, since $\mech'$ is a probability mixture of two scalar and translation invariant, strategyproof mechanisms, it is also scalar and translation invariant and strategyproof.

To evaluate the approximation ratio, we only need to consider the expected egalitarian cost on the profile $(x_1, x_2)= (-1,1)$. The cost for any realized facility location $(a,b)$ is the maximum distance to either agent: $\max(d(x_1, (a,b)), d(x_2, (a,b)))$.

Note that reflecting a point across the $y$-axis swaps its individual distances to $x_1$ and $x_2$:
$$d(x_1, (-a,b)) = d(x_2, (a,b)) \quad \text{and} \quad d(x_2, (-a,b)) = d(x_1, (a,b)).$$

Consequently, the maximum distance to the agents remains identical for a point and its reflection:
$$\max(d(x_1, (-a,b)), d(x_2, (-a,b))) = \max(d(x_1, (a,b)), d(x_2, (a,b))).$$

Taking the expectation over the respective distributions, the expected egalitarian cost under $\mech'$ is equal to the cost under $\mech$ and we conclude $\lambda' = \lambda$.
\end{proof}

We also restrict our lower bound analysis to deviations where $x_1' \le 1$. For any unanimous, translation, and scale-invariant mechanism symmetric with respect to the $y$-axis, any misreport $x_1' \ge 1$ guarantees Agent 1 an expected distance of at least $2$. Since the optimal maximum cost for the reference profile is $1$, this cost induces an approximation ratio of at least $2$, which already exceeds our target lower bound of $\sqrt{2}$. The formal proof of this claim is deferred to Lemma~\ref{lem:right_border_bound} in the appendix.

With these preliminaries in place, we now characterize the distributions that correspond to strategyproof mechanisms:

\begin{lemma} \label{thm:sp_gradient}
A unanimous, translation and scale invariant randomized mechanism $\mech$, with symmetric output distribution $\mathcal{D}$ on the profile $(-1, 1)$, is strategyproof if and only if $\mathcal{D}$ satisfies the following condition, where $p = \mathbb{P}_{\mathcal{D}}((-1,0))$:
\begin{equation}
\label{eq:sp_condition}
-p \;\le\;
\E_{(a,b)\sim\mathcal{D} \setminus \{(-1,0)\}}
\left[
  \frac{1 - a^2 - b^2}{2\sqrt{(a+1)^2 + b^2}}
\right]
\;\le\; p.
\end{equation}
\end{lemma}

\begin{proof}
Because $d(x_1, \mech(x_1',1))= \E_{(a,b)\sim\mathcal{D}}[d(x_1, T_{x_1', 1}(a,b))]$ is a convex function with respect to the report $x_1'$, strategyproofness requires $x_1' = -1$ to be a global minimum. Because of convexity, this is satisfied if and only if $0$ lies in the subdifferential of this expected distance at $x_1' = -1$. 

For a reference point $(a,b)$, the distance from Agent 1's location $x_1 = (-1,0)$ to the mapped facility under misreport $x_1' \le 1$ is:
$$d(x_1, T_{x_1', 1}(a,b)) = \sqrt{\left( \frac{1-x_1'}{2}a + \frac{x_1'+3}{2} \right)^2 + \left( \frac{1-x_1'}{2}b \right)^2}.$$

For any point $(a,b) \neq (-1,0)$, the function $d(x_1, T_{x_1', 1}(a,b))$ is differentiable with respect to $x_1'$ at $x_1' = -1$ with:

\begin{equation}
\label{derivative}
\left. \frac{d}{dx_1'} d\bigl(x_1, T_{x_1', 1}(a,b)\bigr)
\right|_{x_1'=-1}
=
\frac{1 - a^2 - b^2}{2\sqrt{(a+1)^2 + b^2}}.
\end{equation}

For the point $(a,b) = (-1,0)$, the distance simplifies to $|x_1'+1|$, which is not differentiable at $x_1'=-1$ but has a subdifferential interval of $[-1, 1]$. Weighted by its probability $p$, its contribution to the expected subdifferential at point $x_1'=-1$ is $[-p, p]$.

Taking the expectation over the differentiable part $\mathcal{D} \setminus \{(-1,0)\}$ and adding the subdifferential contribution of the point mass at $(-1,0)$, the condition for $0$ to be in the subdifferential at $x_1'=-1$ is:
\begin{equation*}
-p \;\le\;
\E_{(a,b)\sim\mathcal{D} \setminus \{(-1,0)\}}
\left[
  \frac{1 - a^2 - b^2}{2\sqrt{(a+1)^2 + b^2}}
\right]
\;\le\; p.
\end{equation*}

By the $y$-axis symmetry of $\mathcal{D}$, the same logic applies to Agent 2 concluding the proof.
\end{proof}

We now use the characterization from Lemma \ref{eq:sp_condition} to show that the optimal approximation ratio is always achievable by a mechanism supported entirely on the perpendicular bisector of the agent locations. In other words, we establish below that to find the best possible approximation ratio, it suffices to restrict $\mathcal{D}$ to points of the form $(0,b)$ for $b \in \mathbb{R}$.

\begin{lemma} \label{lem:y_axis_projection}
For any unanimous, translation and scale invariant strategyproof 2-agent mechanism $\mech$ with symmetric distribution $\mathcal{D}$ on the profile $(-1,1)$ and approximation ratio $\lambda$, there exists a strategyproof mechanism $\mech'$ with reference distribution $\mathcal{D}'$ supported on the $y$-axis and approximation ratio $\lambda' \leq \lambda$.
\end{lemma}

\begin{proof}
By translation and scale invariance, it is enough to analyze the profile $(x_1, x_2) = (-1, 1)$. Let $p = \mathbb{P}_{\mathcal{D}}((-1,0))$, which by symmetry implies $\mathbb{P}_{\mathcal{D}}((1,0)) = p$. Since the derivative contribution evaluated at $(1,0)$ is zero, Lemma~\ref{thm:sp_gradient} gives us the bound:
\begin{equation}
\label{D_bound}
\left|\E_{(a,b)\sim\mathcal{D} \setminus \{(-1,0), (1,0)\}} \left[ \frac{1 - a^2 - b^2}{2\sqrt{(a+1)^2 + b^2}} \right] \right| \leq p.
\end{equation}
We construct the new distribution $\mathcal{D}'$ on the $y$-axis in two steps.

First, by symmetry, we group the remaining support into pairs of points $\{(a,b), (-a,b)\}$ for $a > 0$ (note each point carries the same probability mass). We map each pair to a single location $(0, b')$ on the $y$-axis with $b' \ge 0$, assigning it their combined probability mass (see Figure~\ref{fig:projection_mapping}). To preserve strategyproofness, we require the expected derivative contribution of the pair given by \eqref{derivative} to equal that of the new mapped location. Multiplying both sides by $-1$ to rearrange the terms, we obtain:
\begin{equation*}
\frac{1}{2}\left[\frac{a^2 + b^2 - 1}{2\sqrt{(a+1)^2 + b^2}} + \frac{(-a)^2 + b^2 - 1}{2\sqrt{(-a+1)^2 + b^2}}\right] = \frac{(b')^2 - 1}{2\sqrt{1 + (b')^2}}.
\end{equation*}
Simplifying, the above is equivalent to:
\begin{equation}\label{eq:projection_identity}
\frac{1}{2}\left[\frac{a^2 + b^2 - 1}{\sqrt{(a+1)^2 + b^2}} + \frac{a^2 + b^2 - 1}{\sqrt{(a-1)^2 + b^2}}\right] = \sqrt{1 + (b')^2} - \frac{2}{\sqrt{1 + (b')^2}}.
\end{equation}

The right-hand side takes the form $f(x) = x - \frac{2}{x}$ for $x = \sqrt{1 + (b')^2} \ge 1$. Note $f(x)$ is strictly increasing on $[1, \infty)$, with range $[-1, \infty)$. The left-hand side of \eqref{eq:projection_identity} is bounded below by $-1$: it is non-negative when $a^2+b^2 \ge 1$, and for $a^2+b^2 < 1$, it is minimized when $b=0$, where it evaluates to $\frac{1}{2}(a^2-1)(\frac{1}{a+1} + \frac{1}{1-a}) = -1$. Thus, a unique solution $x \ge 1$ always exists, determining $b'$.

To show that this transformation does not increase the expected egalitarian cost, observe that under the profile $(-1,1)$, both facility locations $(a,b)$ and $(-a,b)$ share the same maximum cost $\sqrt{(a+1)^2+b^2}$. Since $a > 0$ implies $\sqrt{(a+1)^2+b^2} \ge \sqrt{(a-1)^2+b^2}$, bounding the paired terms on the left-hand side of \eqref{eq:projection_identity} ensures that:
\begin{equation*}
\sqrt{1 + (b')^2} - \frac{2}{\sqrt{1 + (b')^2}} \le \sqrt{(a+1)^2 + b^2} - \frac{2}{\sqrt{(a+1)^2 + b^2}}.
\end{equation*}
Because $f(x)$ is strictly increasing, this inequality implies $\sqrt{1 + (b')^2} \le \sqrt{(a+1)^2 + b^2}$, confirming that the new mapped maximum cost is never higher than the original.

Second, we map the total probability mass $2p$ from the points $(\pm 1,0)$ to a single point $(0,b_0)$ on the $y$-axis. By applying linearity of expectation, the definition of $\mathcal{D}'$, and Lemma ~\ref{eq:sp_condition}, the underlying mechanism $\mech$ is strategyproof in expectation if:
\begin{equation}\label{aaaaa}
\begin{split}
0 &= \E_{(a,b)\sim\mathcal{D}'} \left[ \frac{1 - a^2 - b^2}{2\sqrt{(a+1)^2 + b^2}} \right] \\
  &= 2p \left( \frac{1 - b_0^2}{2\sqrt{1 + b_0^2}} \right) + \E_{(a,b)\sim\mathcal{D} \setminus \{(-1,0), (1,0)\}} \left[ \frac{1 - a^2 - b^2}{2\sqrt{(a+1)^2 + b^2}} \right].
\end{split}
\end{equation}

To solve for $b_0$, let $Z_0 := \sqrt{1 + b_0^2}$ and let $D$ denote the remaining expected value term:
\begin{equation*}
D := \E_{(a,b)\sim\mathcal{D} \setminus \{(-1,0), (1,0)\}} \left[ \frac{1 - a^2 - b^2}{2\sqrt{(a+1)^2 + b^2}} \right].
\end{equation*}
Under these assignments, equality \eqref{aaaaa} is equivalent to $f(Z_0) = \frac{D}{p}$. By~\eqref{D_bound}, we obtain $|D| \leq p$, implying the right-hand side is bounded within $[-1, 1]$. Since $f(1) = -1$, $f(2) = 1$, and $f(x)$ is continuous and strictly increasing, there exists a unique solution $Z_0 \in [1, 2]$, which defines $b_0 \ge 0$. Because the egalitarian cost of the mapped location $(0, b_0)$ is $\sqrt{1+b_0^2} = Z_0 \le 2$, this transformation does not increase the original egalitarian cost of $2$ associated with the points $(\pm 1, 0)$.

Since this construction preserves the pointwise probability weights and ensures that the egalitarian cost contributions never increase, the expected approximation ratio of $\mech'$ under $\mathcal{D}'$ is at most that of $\mech$, and $\mech'$ is strategyproof.
\end{proof}

\begin{figure}[htbp]
    \centering
    \begin{tikzpicture}[scale=1.5]
        \draw[->, gray!60] (-2.5,0) -- (2.5,0) node[right, text=black] {$x$};
        \draw[->, gray!60] (0,-0.5) -- (0,2.5) node[above, text=black] {$y$};

        \fill[black] (-1,0) circle (1.5pt) node[below] {$x_1$};
        \fill[black] (1,0) circle (1.5pt) node[below] {$x_2$};

        \coordinate (A) at (1.2, 1.6);
        \coordinate (Asym) at (-1.2, 1.6);
        \coordinate (B) at (0, 1.81); 

        \draw[dashed, thick, blue!60] (-1,0) -- (A);
        \draw[dashed, thick, blue!60] (1,0) -- (Asym);
        \draw[dashed, thick, red!60] (-1,0) -- (B);

        \draw[->, shorten >=3pt, shorten <=3pt, gray, semithick, dotted] (A) to[bend right=15] (B);
        \draw[->, shorten >=3pt, shorten <=3pt, gray, semithick, dotted] (Asym) to[bend left=15] (B);

        \fill[black] (A) circle (1.5pt) node[right] {$(a,b)$};
        \fill[black] (Asym) circle (1.5pt) node[left] {$(-a,b)$};
        \fill[black] (B) circle (1.5pt) node[above right] {$(0,b')$};
    \end{tikzpicture}
    \caption{Symmetric facility locations $(a,b)$ and $(-a,b)$ are mapped onto the perpendicular bisector at $(0,b')$. The dashed lines indicate the maximum egalitarian cost of the facilities under the profile $(-1,1)$.}
    \label{fig:projection_mapping}
\end{figure}
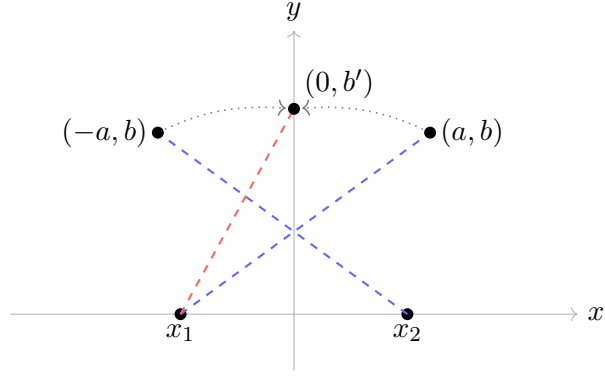

This reduction allows us to establish the desired lower bound:

\begin{theorem}
Any strategyproof, translation invariant, and scale invariant randomized mechanism $\mech$ has an approximation ratio of at least $\sqrt{2}$.
\end{theorem}

\begin{proof}
We first establish the lower bound for the case of $n = 2$ agents.

By Lemma \ref{lem:y_axis_projection}, it suffices to consider mechanisms whose output distribution $\mathcal{D}$ on the profile $\vec{x} = (-1, 1)$ is supported entirely on the $y$-axis, meaning any realized facility location is of the form $(0, b) \sim \mathcal{D}$.

Because the agents are located at $x_1 = (-1,0)$ and $x_2 = (1,0)$, their distances to any point on the $y$-axis are identical. We define the random variable $Z := d(x_1, \mech(\vec{x})) = \sqrt{1 + b^2}$ to represent this distance. Note that its expectation $\E_{\mathcal{D}}[Z]$ corresponds to the expected egalitarian cost. Rearranging this identity we obtain $b^2 = Z^2 - 1$.

Applying Lemma \ref{thm:sp_gradient} with $a=0$, the strategyproofness condition simplifies to:
$$\E_{\mathcal{D}} \left[ \frac{1 - b^2}{2\sqrt{1 + b^2}} \right] = 0$$

Substituting $b^2 = Z^2 - 1$ and $\sqrt{1+b^2} = Z$ yields:
$$\E_{\mathcal{D}} \left[ \frac{1 - (Z^2 - 1)}{2Z} \right] = \E_{\mathcal{D}} \left[ \frac{2 - Z^2}{2Z} \right] = 0$$

By linearity of expectation, this expression separates into:
$$\E_{\mathcal{D}} \left[ \frac{1}{Z} - \frac{Z}{2} \right] = 0 \implies \E_{\mathcal{D}}[Z] = 2\E_{\mathcal{D}}\left[\frac{1}{Z}\right]$$

Since $f(t) = 1/t$ is convex for $t > 0$, Jensen's inequality implies $\E_{\mathcal{D}}\left[\frac{1}{Z}\right] \ge \frac{1}{\E_{\mathcal{D}}[Z]}$. Substituting this inequality gives us:
$$\E_{\mathcal{D}}[Z] \ge \frac{2}{\E_{\mathcal{D}}[Z]}$$

This implies $ \SC(\vec{x},\mech) = \E_{\mathcal{D}}[Z] \ge \sqrt{2}$. Because the optimal maximum cost for this profile is $\opt = 1$, the expected approximation ratio is at least $\sqrt{2}$.

By Lemma \ref{lem:population_extension2}, the lower bound extends to an arbitrary number of agents.

\end{proof}

\subsection{Circle Setting}

We now consider the setting where the agents' locations and reports are restricted to the unit circle $\mathcal{I} = S^1 \subset \mathbb{R}^2$ centered at the origin $O = (0,0)$, while the facility may be placed anywhere in the plane, $\mathcal{O} = \mathbb{R}^2$. 

Depending on the context, we represent an agent's location either by its Cartesian coordinates $(x,y) \in \mathbb{R}^2$ or by its angular coordinate $\theta \in [0, 2\pi)$. The angle $\theta$ of a report measures the counterclockwise rotation from the positive $x$-axis to the line connecting the origin to the agent's location, $(x,y) = (\cos\theta, \sin\theta)$. This is illustrated in Figure \ref{fig:angle_definition} below. 

\begin{figure}[htbp]
    \centering
    \begin{tikzpicture}[scale=1.8]
        \draw[->, thick, gray!70] (-1.5,0) -- (1.5,0) node[right, text=black] {$x$};
        \draw[->, thick, gray!70] (0,-1.5) -- (0,1.5) node[above, text=black] {$y$};
        
        \draw[thick, dashed, gray!80] (0,0) circle (1);
        
        \filldraw[black] (0,0) circle (1pt) node[below left] {$O$};
        
        \coordinate (A) at (40:1);
        \draw[thick, dashed, gray!70] (0,0) -- (A);
        \filldraw[black] (A) circle (1.5pt) node[right=4pt, text=black] { $x_1 = (\cos\theta, \sin\theta)$};
        
        \draw[->, thick, black] (0.35,0) arc (0:40:0.35) node[midway, right=3pt] {$\theta$};
        
        \coordinate (B) at (150:1);
        \draw[thick, dashed, gray!70] (0,0) -- (B);
        \filldraw[black] (B) circle (1.5pt) node[above left=1pt, text=black] {$x_2$};
        
        \coordinate (Out) at (105:1.3);
        \filldraw[black] (Out) circle (1.5pt) node[above left, xshift = -0pt] {$\mech(\vec{x})$};
    \end{tikzpicture}
    \caption{Illustration of the circle setting.}
    \label{fig:angle_definition}
\end{figure}
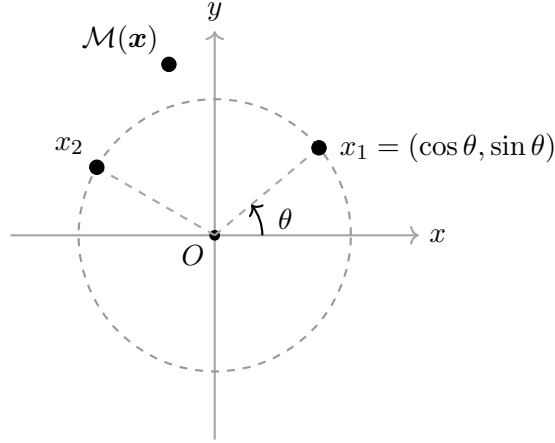

As the main result for this setting, we will present a randomized mechanism that has an approximation ratio of $1.5$.

\begin{theorem} \label{thm:circle:main}
    There exists a randomized mechanism $\mech$ that is group-strategyproof in expectation and achieves an approximation ratio of $1.5$\footnote{Note that the mechanism of the theorem is also strongly group-strategyproof.}.
\end{theorem}

\subsubsection{Preliminaries}

To build the intuition for our mechanism, we first establish how the spread of the agents around the circle influences the optimal facility location. This relies on the concept of the \emph{minimum spanning arc} and its bounding agents.

\begin{definition}[Minimum Spanning Arc and Extreme Agents]
    Let $\vec{x}$ be a profile of agents on $S^1$. A \emph{spanning arc} is any continuous arc on the circle that contains the locations of all agents in $\vec{x}$. The \emph{minimum spanning arc} is a spanning arc of minimal length, and we denote its length by $\alpha \in [0, 2\pi)$. When $\alpha < \pi$, this minimum spanning arc is unique. In such cases, the \emph{extreme agents} (or extreme points) of $\vec{x}$ are the two agents located at the endpoints of this arc.
\end{definition}

For a given profile of reports, we say an arc spans $[\alpha, \beta]$, with $\alpha\leq \beta$, if it corresponds to the minimum spanning arc that includes the locations at angles $\alpha$ and $\beta$. 

For any profile $\vec{x}$, the optimal facility location $z^*$ and its corresponding minimax cost $\opt$ depend on whether this minimum spanning arc spans an angle smaller or greater than $\pi$:

\begin{itemize}
    \item \textbf{Large Arc ($\alpha \ge \pi$):} In this case, the agents cannot be contained within a single semicircle. The smallest enclosing circle containing all agents is the unit circle itself. The optimal location for the facility is the origin $O$. The distance to any agent on $S^1$ is $1$, yielding $\opt = 1$. See Figure \ref{fig:large_prel}.
    
    \item \textbf{Small Arc ($\alpha < \pi$):} All agents are concentrated on one side of the circle, bounded by two extreme agents with locations $A$ and $B$. The smallest enclosing circle of this profile has the chord $AB$ as its diameter The optimal facility location $z^*$ is the midpoint of this chord. The length of the chord is $2\sin(\alpha/2)$ and the maximum distance to any agent is the radius of this enclosing circle, yielding an optimal cost of $\opt = \sin(\alpha/2)$. See Figure \ref{fig:small_prel}
\end{itemize}

\begin{figure}[h!]
    \centering
    \begin{subfigure}{0.48\textwidth}
        \centering
        \begin{tikzpicture}[scale=1.5]
            \draw[ultra thick, blue, dashed] (0,0) circle (1.5);
            
            \filldraw[red] (0,0) circle (2pt) node[below right, text=black] {$O$};
            
            \filldraw[black] (20:1.5) circle (2pt);
            \filldraw[black] (130:1.5) circle (2pt);
            \filldraw[black] (250:1.5) circle (2pt);
            
            \draw[thick] (0,0) -- (130:1.5);
        \end{tikzpicture}
        \caption{Large Arc: $\opt = 1$}
        \label{fig:large_prel}
    \end{subfigure}\hfill
    \begin{subfigure}{0.48\textwidth}
        \centering
        \begin{tikzpicture}[scale=1.5]
            \draw[thick, dashed, gray!60] (0,0) circle (1.5);
            \filldraw[gray] (0,0) circle (1.5pt) node[below left=2pt, text=black] {$O$};
            
            \coordinate (A) at (100:1.5);
            \coordinate (B) at (0:1.5);
            \coordinate (Z) at ($(A)!0.5!(B)$); 
            \coordinate (Origin) at (0,0);
            
            \draw[ultra thick, blue] (0:1.5) arc (0:100:1.5);
            
            \draw[gray, dashed] (Origin) -- (A);
            \draw[gray, dashed] (Origin) -- (B);
            \draw[gray, dashed] (Origin) -- (Z);
            
            \coordinate (U) at ($(Z)!4pt!(Origin)$);
            \coordinate (V) at ($(Z)!4pt!(B)$);
            \coordinate (W) at ($(U)+(V)-(Z)$);
            \draw[gray, thick] (U) -- (W) -- (V);

            \draw[gray, ->, >=stealth] (0:0.3) arc (0:100:0.3) node[pos=0.85, above right=-3pt, scale=0.9, text=black] {$\alpha$};
            \draw[gray] (0:0.6) arc (0:50:0.6) node[midway, right=1pt, scale=0.8, text=black] {$\frac{\alpha}{2}$};
            
            \draw[thick, gray!80] (A) -- (B);
            \draw[dashed, blue!50, thick] (Z) circle ({1.5*sin(50)});
            
            \filldraw[black] (A) circle (2pt) node[above left] {$A$};
            \filldraw[black] (B) circle (2pt) node[below right] {$B$};
            \filldraw[black] (30:1.5) circle (1.5pt);
            \filldraw[black] (70:1.5) circle (1.5pt);
            
            \filldraw[red] (Z) circle (2pt) node[above right=1pt, text=black] {$z^*$};
            
            \draw[thick] (Z) -- (B);
        \end{tikzpicture}
        \caption{Small Arc: $\opt = \sin(\alpha/2)$}
        \label{fig:small_prel}
    \end{subfigure}
    
    \caption{Geometric configurations of the described cases. The optimal output locations are marked in red.}
    \label{fig:opt_preliminaries}
\end{figure}
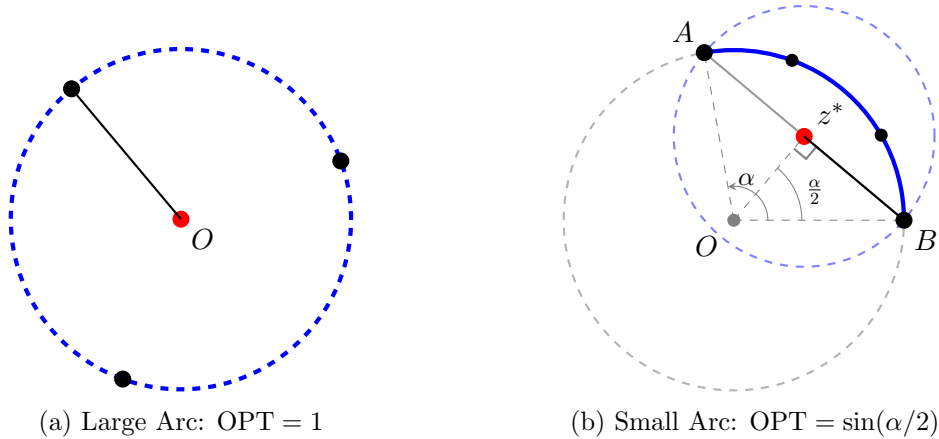

\begin{remark}
The deterministic mechanism that, for a given profile $\vec{x}$ always outputs the optimal egalitarian location is not strategyproof. In the small arc case, similar to the line setting, extreme agents can misreport their locations to pull the chord midpoint closer to themselves.
\end{remark}
\subsubsection{The Chord-Midpoint Mechanism}

To achieve the guarantees stated in Theorem \ref{thm:circle:main}, we introduce the Chord-Midpoint Mechanism. This randomized mechanism is described below, in Algorithm \ref{mech:chord:midpoint}, and illustrated in Figure \ref{fig:chord_midpoint_mech}.

\begin{algorithm}[h!]
    \KwIn{Strategic profile $\vec{x}$ of reported agent locations}
    Let $A$ and $B$ denote the extreme agents for the reported locations $\vec{x}$\;
    Let $\alpha$ denote the angle of the minimum spanning arc defined by $A$ and $B$\;    
    \eIf{$\alpha \ge \pi$}
    {
        Output the origin $O = (0,0)$ with probability one\;
    }{
        Compute $\lambda(\alpha) =  \frac{\cos(\alpha/2)}{2(1 + \cos(\alpha/2))}$\;
        Compute the midpoint $z$ of the chord between $A$ and $B$\;
        Output $A$ and $B$ with prob.~$\lambda(\alpha)$ each, and $z$ with prob.~$1-2\lambda(\alpha)$\;
    }
    \caption{Chord-Midpoint Mechanism $\mech$}
    \label{mech:chord:midpoint}
\end{algorithm}

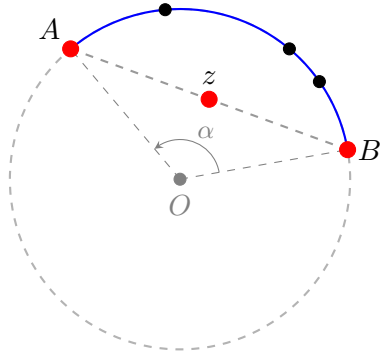
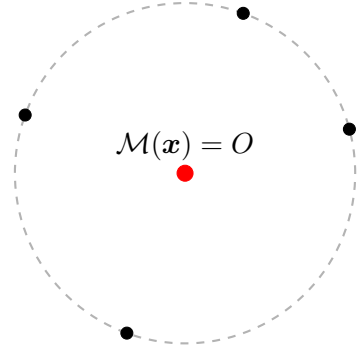
\begin{figure}[htbp]
    \centering
    \begin{subfigure}{0.48\textwidth}
        \centering
        \begin{tikzpicture}[scale=1.5]
            \draw[thick, dashed, gray!60] (0,0) circle (1.5);
            \filldraw[gray] (0,0) circle (1.5pt) node[below=2pt] {$O$};
            
            \coordinate (A) at (130:1.5);
            \coordinate (B) at (10:1.5);
            \coordinate (Z) at ($(A)!0.5!(B)$); 
            
            \draw[thick, blue] (10:1.5) arc (10:130:1.5);
            \draw[gray, dashed] (0,0) -- (A);
            \draw[gray, dashed] (0,0) -- (B);
            \draw[gray, ->, >=stealth] (10:0.35) arc (10:130:0.35) node[midway, above right=-2pt, scale=0.9] {$\alpha$};
            \draw[thick, gray!80, dashed] (A) -- (B);
            
            \filldraw (35:1.5) circle (1.5pt);
            \filldraw (50:1.5) circle (1.5pt);
            \filldraw (95:1.5) circle (1.5pt);
            
            \filldraw[red] (A) circle (2pt) node[above left, text=black] {$A$};
            \filldraw[red] (B) circle (2pt) node[right, text=black] {$B$};
            \filldraw[red] (Z) circle (2pt) node[above, yshift=2pt, text=black] {$z$};
        \end{tikzpicture}
        \caption{Small Arc ($\alpha < \pi$)}
        \label{fig:circle_small_arc}
    \end{subfigure}
    \hfill
    \begin{subfigure}{0.48\textwidth}
        \centering
        \begin{tikzpicture}[scale=1.5]
            \draw[thick, dashed, gray!60] (0,0) circle (1.5);
            
            \filldraw (15:1.5) circle (1.5pt);
            \filldraw (70:1.5) circle (1.5pt);
            \filldraw (160:1.5) circle (1.5pt);
            \filldraw (250:1.5) circle (1.5pt);
            
            \filldraw[red] (0,0) circle (2pt) node[above=1pt, text=black] {$\mech(\vec{x}) = O$};
        \end{tikzpicture}
        \caption{Large Arc ($\alpha \ge \pi$)}
        \label{fig:circle_large_arc}
    \end{subfigure}
    \caption{Illustration of the Chord-Midpoint Mechanism. Candidate output locations are shown in red.}
    \label{fig:chord_midpoint_mech}
\end{figure}

The key idea of our mechanism, similar to the LRM Mechanism for the classic line setting, is to randomize between the extreme agents and the optimal location according to a probability distribution. However, the main difference is that to account for the geometry of the circle, our probability distribution depends on the angle of the minimum spanning arc, allowing it to transition smoothly into the large arc case where $\alpha \ge \pi$.

The following lemma establishes a fine-grained performance guarantee that depends on the angle $\alpha$ of the spanning arc between the two extreme agents. As $\alpha \to \pi$, the approximation ratio approaches $1$, whereas as $\alpha \to 0$, it approaches $1.5$. Intuitively, this can be explained by the observation that our mechanism converges to the optimal solution as $\alpha$ increases and to the LRM Mechanism as $\alpha$ decreases.
\begin{lemma} \label{thm:approx_ratio}The Chord-Midpoint Mechanism $\mech$ has an approximation ratio of $1.5$. More precisely, for any profile $\vec{x}$ that spans an angle $\alpha$, we have:
$$\frac{\SC(\mech, \vec{x})}{\opt} = \max\left\{\frac{1 + 2\cos(\alpha/2)}{1+\cos(\alpha/2)}, 1\right\}$$
\end{lemma}

\begin{proof}
    We rely on the optimal configurations established in the preliminaries and distinguish between the two configurations of the spanning arc.

  \begin{itemize}
        \item \textbf{Case 1 ($\alpha \ge \pi$):} The profile spans at least a semicircle, so $\opt = 1$. The mechanism outputs the origin, yielding a distance of $1$ for all agents $x_i \in S^1$. Thus, $\SC(\mech,\vec{x}) = 1 \le 1.5 \cdot \opt$.
        
        \item \textbf{Case 2 ($\alpha < \pi$):} Let A and B be the locations of the extreme agents of the arc that span an angle of $\alpha$ and let $z$ be the midpoint of the arc between them. The optimal cost is $\opt = \sin(\alpha/2)$, and the mechanism outputs a facility $Y \in \{A, B, z\}$. Because all agents lie inside the arc bounded by $A$ and $B$, the maximum distance $d^*$ from $Y$ to any agent is achieved at one of the endpoints. If $Y = A$ or $Y = B$, then $d^* = \|A-B\| = 2\sin(\alpha/2)$. If $Y = z$, then $d^* = \|z-A\| = \sin(\alpha/2)$.
    \end{itemize}
    
 Taking the expectation over the mechanism's randomness, the expected maximum cost is:
    \begin{align*}
        \SC(\mech,\vec{x}) &= \lambda(\alpha)\cdot(2\sin(\alpha/2)) + \lambda(\alpha)\cdot(2\sin(\alpha/2)) + (1-2\lambda(\alpha))\cdot\sin(\alpha/2) \\
        &= 4\lambda(\alpha)\cdot\sin(\alpha/2) + \sin(\alpha/2) - 2\lambda(\alpha)\cdot\sin(\alpha/2) \\
        &= \sin(\alpha/2)\cdot(1 + 2\lambda(\alpha)) \\
        &= \sin(\alpha/2) \cdot\left[ 1 + \frac{\cos(\alpha/2)}{1+\cos(\alpha/2)} \right] \\
        &= \sin(\alpha/2)\cdot \left[ \frac{1 + 2\cos(\alpha/2)}{1+\cos(\alpha/2)} \right].
    \end{align*}
    
  Thus, the approximation ratio is $\frac{\SC(\mech,\vec{x})}{\opt} = \frac{1 + 2\cos(\alpha/2)}{1+\cos(\alpha/2)}$. For $\alpha \in (0, \pi)$, this expression is decreasing in $\alpha$ so:
$$
\frac{\SC(\mech,\vec{x})}{\opt} \le \lim_{\alpha \to 0} \frac{1 + 2\cos(\alpha/2)}{1+\cos(\alpha/2)} = \frac{3}{2}.
$$
\end{proof}

\subsubsection{Strategyproofness}

The entirety of this subsection is dedicated to proving that the Chord-Midpoint Mechanism is strategyproof in expectation. This result will later serve as a building block to prove the stronger property of group-strategyproofness.

\begin{theorem} \label{thm:strategyproof}
    The Chord-Midpoint Mechanism $\mech$ is strategyproof in expectation.
\end{theorem}

To prove Theorem \ref{thm:strategyproof}, we analyze all possible individual deviations. Because the full proof is long and relies on a series of technical lemmas, we first provide an overview by categorizing these deviations into three exhaustive types based on how a misreport alters the geometry of the minimum spanning arc (see Figure~\ref{fig:deviation_types}):
\begin{enumerate}
   \item \textbf{Arc Expansion:} An agent reports a location outside the truthful minimum spanning arc, increasing the angle $\alpha$. This deviation either transitions the profile into the large arc case ($\alpha \ge \pi$) or extends one side of the arc while leaving the other untouched. We prove the stronger result that this unilateral deviation is not profitable and that it increases the expected cost for \textit{all} agents truthfully located inside the original arc.
   \item \textbf{Arc Shrinkage} An extreme agent reports a location inside the true spanning arc, decreasing the angle $\alpha$. We prove this deviation penalizes the extreme agent who misreported.
   \item \textbf{Cross-boundary Deviations:} An extreme agent deviates past the opposite extreme point. This causes one side of the arc to extend while the other shrinks. We prove this is not profitable by showing that it is equivalent to a sequential expansion and shrink, which results in a net penalty for the deviator.
\end{enumerate}

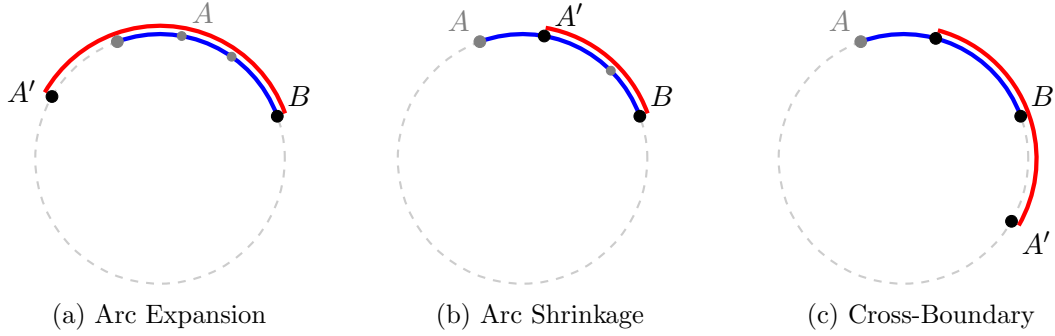
\begin{figure}[htbp]
    \centering
    \begin{subfigure}{0.31\textwidth}
        \centering
        \begin{tikzpicture}[scale=1.1]
            \draw[thick, dashed, gray!40] (0,0) circle (1.5);
            
            \draw[ultra thick, blue] (20:1.5) arc (20:110:1.5);
            \draw[ultra thick, red] (20:1.60) arc (20:150:1.60);
            
            \filldraw[black] (20:1.5) circle (2pt) node[above right] {$B$};
            \filldraw[gray] (110:1.5) circle (2pt);
            \filldraw[black] (150:1.5) circle (2pt) node[left, xshift=-2pt, yshift=2pt] {$A'$};
            
            \filldraw[gray] (55:1.5) circle (1.5pt);
            
            \filldraw[gray] (80:1.5) circle (1.5pt)node[above right, text=gray, yshift=1pt] {$A$};
        \end{tikzpicture}
        \caption{Arc Expansion}
        \label{fig:dev_expansion}
    \end{subfigure}
    \hfill
    \begin{subfigure}{0.31\textwidth}
        \centering
        \begin{tikzpicture}[scale=1.1]
            \draw[thick, dashed, gray!40] (0,0) circle (1.5);
            
            \draw[ultra thick, blue] (20:1.5) arc (20:110:1.5);
            \draw[ultra thick, red] (20:1.60) arc (20:80:1.60);
            
            \filldraw[black] (20:1.5) circle (2pt) node[above right] {$B$};
            \filldraw[gray] (110:1.5) circle (2pt) node[above left, text=gray] {$A$};
            \filldraw[black] (80:1.5) circle (2pt) node[above right] {$A'$};
            
            \filldraw[gray] (45:1.5) circle (1.5pt);
        \end{tikzpicture}
        \caption{Arc Shrinkage}
        \label{fig:dev_shrinking}
    \end{subfigure}
    \hfill
    \begin{subfigure}{0.31\textwidth}
        \centering
        \begin{tikzpicture}[scale=1.1]
            \draw[thick, dashed, gray!40] (0,0) circle (1.5);
            
            \draw[ultra thick, blue] (20:1.5) arc (20:110:1.5);
            \draw[ultra thick, red] (75:1.60) arc (75:-30:1.60);
            
            \filldraw[gray] (110:1.5) circle (2pt) node[above left, text=gray] {$A$};
            \filldraw[black] (75:1.5) circle (2pt) node[above]
            {}; 
            \filldraw[black] (20:1.5) circle (2pt) node[above right] {$B$};
            \filldraw[black] (-30:1.5) circle (2pt) node[below right] {$A'$};
        \end{tikzpicture}
        \caption{Cross-Boundary}
        \label{fig:dev_cross_boundary}
    \end{subfigure}
    
    \caption{The three types of deviations based on modifications to the minimum spanning arc. In each case, the blue arc corresponds to the truthful minimum spanning arc bounded by extreme locations $A$ and $B$, while the red arc outlines the minimum spanning arc resulting from the misreport $A'$.}
    \label{fig:deviation_types}
\end{figure}

For the remainder of this subsection, whenever a profile corresponds to a minimum spanning arc of angle $\alpha < \pi$, we interchangeably use $A$ and $B$ to denote the extreme agents and their locations. By rotational symmetry of the mechanism, we assume without loss of generality that the angular coordinate of $B$ is $0$ and that of $A$ is $0\leq \alpha \leq \pi$. Consequently, any interior agent $x_i$ can be characterized by an angular coordinate $\beta \in [0, \alpha]$.

When an agent misreports, the change in the spanning arc shifts the position of the chord midpoint. To track how this shift impacts an agent's expected cost, we define the function $D(\theta, \phi)$ as the distance between an agent at angle $\theta$ and the chord midpoint of an arc spanning $[0, \phi]$.

We now present the technical lemmas required to evaluate these strategic deviations. Their proofs rely primarily on algebraic and trigonometric manipulations and are deferred to the appendix.

The first two lemmas decompose the distance function $D(\theta, \phi)$ and rewrite it in an alternative form.

\begin{restatable}{lemma}{lemDistanceFactorization}
    \label{lem:distance_factorization}
    For an agent at angle $\theta \in [0,\pi]$ and a spanning arc $[0, \phi]$ where $\phi \in [0, \pi]$, the squared distance $D(\theta, \phi)^2$ factors as:
    \begin{align*}
    \scalebox{1.08}{$D(\theta, \phi)^2 = \left( \sin\left(\frac{\theta}{2}\right) - \cos\left(\frac{\phi}{2}\right)\sin\left(\frac{\phi-\theta}{2}\right) \right)^2 + \left( \cos\left(\frac{\theta}{2}\right) - \cos\left(\frac{\phi}{2}\right)\cos\left(\frac{\phi-\theta}{2}\right) \right)^2.$}
    \end{align*}
\end{restatable}

\begin{restatable}{lemma}{lemInternalIdentities}
    \label{lem:internal_identities}
    For an agent at angle $\theta \in [0,\pi]$ and a spanning arc $[0, \phi]$ where $\phi \in [0, \pi]$, the squared distance $D(\theta, \phi)^2$ to the chord midpoint satisfies:
    \[
    D(\theta, \phi)^2 = g(\theta, \phi)^2 + \left(1 + \cos\frac{\phi}{2}\right)h(\theta, \phi),
    \]
    where
    \[
    g(\theta, \phi) = \sin\frac{\theta}{2} - \cos\frac{\phi}{2} \sin\frac{\phi - \theta}{2} \quad \text{and} \quad h(\theta, \phi) = \left(1 - \cos\frac{\phi}{2}\right)\sin^2\frac{\phi - \theta}{2}.
    \]
\end{restatable}

Using these identities, we obtain an expression for the expected cost of an agent when the reported profile spans $[0, \phi]$ where $\phi \leq \pi$.

\begin{restatable}{lemma}{lemInternalCostSimplification}
\label{lem:internal_cost_simplification}
    The expected cost for an agent at angle $\theta \in [0,\pi]$ under a manipulated profile $\vec{x'}$ spanning $[0, \phi]$ is:
    \[
    \E[d(\mech(\vec{x'}),x_i)] = \sin\left(\frac{\theta}{2} \right) + \frac{D(\theta, \phi) - g(\theta, \phi)}{1 + \cos(\frac{\phi}{2})}.
    \]
\end{restatable}

Finally, we establish that the second term in the formula above is strictly increasing with respect to the reported spanning angle $\phi$.

\begin{restatable}{lemma}{lemInternalMonotonicity}
    \label{lem:internal_monotonicity}
    The function $w(\phi) = \frac{D(\theta, \phi) - g(\theta, \phi)}{1 + \cos(\phi/2)}$ is strictly increasing with respect to the spanning angle $\phi$ for $\phi \in [\theta, \pi)$.
\end{restatable}
With these technical lemmas established, we now evaluate each of the three deviation types outlined in our roadmap to prove that no misreport is profitable.

\paragraph*{Case 1: Arc Expansion}

We first consider scenarios where an agent reports a location outside the true minimum spanning arc, expanding the angle $\alpha$ to a larger angle $\gamma$. We show that this expansion increases the expected distance for every agent inside the original arc.

\begin{lemma}
    \label{lem:expanded:arc}
   Let $\vec{x}$ be a truthful profile that spans $[0,\alpha]$ with $\alpha < \pi$ and let $\vec{x'}$ be a profile that spans $[0,\gamma]$ with $\alpha < \gamma < \pi$ or a profile with spanning angle $\ge \pi$. Then, $\E[d(x_i,\mech(\vec{x}))] < \E[d(x_i,\mech(\vec{x'}))]$ for all $i \in [n]$.
\end{lemma}

\begin{proof}
    Fix an arbitrary agent $x_i$ at true angle $\beta \in [0, \alpha]$ and assume that the manipulated profile $\vec{x'}$ spans the arc $[0,\gamma]$ with $\gamma < \pi$. By Lemma \ref{lem:internal_cost_simplification}, the expected distance of $x_i$ for profile $\vec{x'}$ is 
    $$
    \E[d(\mech(\vec{x'}),x_i)] = \sin\left(\frac{\beta}{2}\right) + \frac{D(\beta, \gamma) - g(\beta, \gamma)}{1+\cos(\gamma/2)}. 
    $$
    Applying the same logic for the truthful profile $\vec{x}$ gives
    $$
    \E[d(\mech(\vec{x}),x_i)] = \sin\left(\frac{\beta}{2}\right) + \frac{D(\beta, \alpha) - g(\beta, \alpha)}{1+\cos(\alpha/2)}. 
    $$
   Since $\alpha < \gamma < \pi$, by Lemma \ref{lem:internal_monotonicity}, $\E[d(\mech(\vec{x'})), x_i] > \E[d(\mech(\vec{x})), x_i]$.

    If $\vec{x'}$ instead spans an arc of angle $\ge \pi$, then $\E[d(\mech(\vec{x'}),x_i)] = 1$ because the mechanism outputs the origin. However, as shown in Lemma \ref{lem:internal_monotonicity}, the expected distance is strictly increasing with the span angle. The limit as the span approaches $\pi$ evaluates to one. Since $\alpha < \pi$ by assumption, this implies $\E[d(\mech(\vec{x}),x_i)] < 1$. 
\end{proof}

\paragraph*{Case 2: Arc Shrinkage}

We evaluate shrinking deviations across two scenarios based on the initial angle of the spanning arc. We begin with the case where the truthful profile spans at least a semicircle ($\alpha \ge \pi$).

\begin{lemma}
    \label{lem:removing:center}
    Let $\vec{x}$ be a profile with a minimum spanning arc of angle $\alpha \ge \pi$. Let $\vec{x'} = (x_i',\vec{x}_{-i})$ be such that the angle of the minimum spanning arc is $\theta < \pi$. Then $\E[d(\mech(\vec{x}),x_i)] < \E[d(\mech(\vec{x'}),x_i)]$.
\end{lemma}

\begin{proof}
    Profile $\vec{x}$ obtains an expected egalitarian cost of $1$. Without loss of generality, assume the profile $\vec{x'}$ spans the arc $W' = [-\theta/2, \theta/2]$. Let $z' = (\cos(\theta/2), 0)$ be the chord midpoint of its extreme agents. 
    
   Because the truthful profile $\vec{x}$ has a span $\ge \pi$ and all non-deviating agents lie within $W'$, the true location $x_i$ must correspond to an angle $\gamma \in [\pi - \theta/2, \pi + \theta/2]$ to satisfy the true spanning condition. Consequently, we get 
   
   \begin{equation}
   \label{cosineq}
   \cos\gamma \le \cos(\pi - \theta/2) =-\cos(\theta/2).
   \end{equation}

  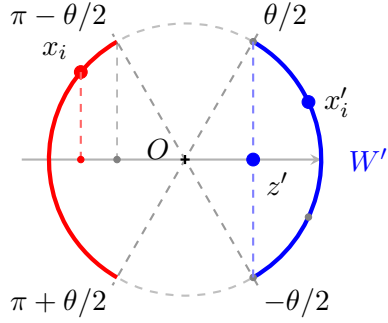
\begin{figure}[htbp]
    \centering
    \begin{tikzpicture}[scale=1.2]
        
        \draw[thick, dashed, gray!50] (0,0) circle (1.5);
        
        \draw[thick, gray!60, ->, >=stealth] (-1.8, 0) -- (1.5, 0);
        
        \draw[gray!80, dashed, thick] (60:1.6) -- (240:1.6);
        \draw[gray!80, dashed, thick] (-60:1.6) -- (120:1.6);
        
        \draw[thick, black] (-0.05, 0) -- (0.05, 0);
        \draw[thick, black] (0, -0.05) -- (0, 0.05);
        \node[above=4pt, left=2pt] at (0,0) {$O$};
        
        \draw[ultra thick, blue] (-60:1.5) arc (-60:60:1.5) node[midway, right=6pt] {$W'$};
        
        \filldraw[gray] (60:1.5) circle (1.0pt) node[above right, text=black] {$\theta/2$};
        \filldraw[gray] (-60:1.5) circle (1.0pt) node[below right, text=black] {$-\theta/2$};
        \filldraw[gray] (-25:1.5) circle (1.0pt);
        
        \coordinate (Xiprime) at (25:1.5);
        \filldraw[blue] (Xiprime) circle (2pt) node[right=2pt, text=black] {$x_i'$};
        
        \coordinate (Zprime) at ({1.5*cos(60)}, 0);
        \draw[thick, blue!50, dashed] (60:1.5) -- (Zprime);
        \draw[thick, blue!50, dashed] (-60:1.5) -- (Zprime);
        \filldraw[blue] (Zprime) circle (2pt) node[below right=1pt, text=black] {$z'$};
        
        \draw[ultra thick, red] (120:1.5) arc (120:240:1.5);
        
        \node[above left] at (120:1.5) {$\pi - \theta/2$};
        \node[below left] at (240:1.5) {$\pi + \theta/2$};
        
        \coordinate (MinusZprime) at ({1.5*cos(120)}, 0);
        \draw[gray, dashed] (120:1.5) -- (MinusZprime);
        \filldraw[gray] (MinusZprime) circle (1.0pt);
        
        \coordinate (Xi) at (140:1.5);
        \filldraw[red] (Xi) circle (2pt) node[above left=1pt, text=black] {$x_i$};
        
        \coordinate (XiProj) at ({1.5*cos(140)}, 0);
        \draw[red!60, dashed, thick] (Xi) -- (XiProj);
        \filldraw[red] (XiProj) circle (1.0pt);

    \end{tikzpicture}
    \caption{Geometric configuration for Lemma \ref{lem:removing:center}. The red arc shows the complementary region $[\pi - \theta/2, \pi + \theta/2]$ where the true location $x_i$ must lie.}
    \label{fig:lemma_removing_center}
\end{figure}

    Let $f(X) = \|x_i - X\|$ be the distance from the mechanism's output $X$ to $x_i$. Since $f$ is convex, by Jensen's inequality:
   $\E[d(\mech(\vec{x'}), x_i)] = \E[\|x_i - X\|] \geq \|x_i - \E[X]\| = \|x_i - z'\|$.
    Applying the Law of Cosines to the triangle formed by the origin $O$, the agent location $x_i$, and the chord midpoint $z'$, and using $\|x_i\|=1$ and $\|z'\|=\cos(\theta/2)$ we obtain:
\begin{align*}
\|x_i - z'\|^2 &= 1^2 + \cos^2(\theta/2) - 2\cdot\cos(\theta/2)\cdot\cos(\gamma) \\
&\ge 1 + \cos^2(\theta/2) - 2\cdot\cos(\theta/2)\cdot(-\cos(\theta/2)) \\
&= 1 + 3\cdot\cos^2(\theta/2) > 1,
\end{align*}
where we used \ref{cosineq} for the inequality step.

Thus, $\E[d(\mech(\vec{x'}), x_i)] > 1 = \E[d(\mech(\vec{x}), x_i)]$.
\end{proof}

Next, we consider the scenario where the true arc satisfies $\alpha < \pi$ and an extreme agent reports a location strictly inside the spanning arc.
\begin{lemma}
    \label{lem:shrinking_bound}
    Let $\vec{x}$ be a profile with a minimum spanning arc of angle $\alpha < \pi$. Let $\vec{x'} = (x_i',\vec{x}_{-i})$ be such that the angle of the minimum spanning arc is $\gamma < \alpha$. Then $\E[d(\mech(\vec{x}),x_i)] < \E[d(\mech(\vec{x'}),x_i)]$. 
\end{lemma}

\begin{proof}
Since the deviation shrinks the minimum spanning arc, it must be done by one of the extreme points. By our established setup, the stationary extreme $B$ is located at angle $0$, agent $A$ is at angle $\alpha$. Without loss of generality, let $A$ be the deviating agent to a location $A'$. Thus, the manipulated profile $\vec{x'}$ spans the smaller arc $[0, \gamma]$ with the new extreme $A'$ at angle $\gamma$. 

The expected cost for agent $A$ under the truthful profile $\vec{x}$ is $\E[d(A,\mech(\vec{x}))] = 2\lambda(\alpha) \cdot\sin(\alpha/2) + (1-2\lambda(\alpha))\cdot \sin(\alpha/2) = \sin(\alpha/2)$. Therefore, it only remains to show that $\E[d(\mech(\vec{x'}),A)] > \sin(\alpha/2)$.

Let $z'$ be the chord midpoint of the reported arc $[0, \gamma]$. Given that agent $A$ is at angle $\alpha$ and the new reported extreme $A'$ is at angle $\gamma$, the distances from $A$ to the possible output locations $B$, $A'$, and $z'$ are $2\sin(\alpha/2)$, $2\sin\left(\frac{\alpha - \gamma}{2}\right)$, and $D(\alpha, \gamma)$, respectively. Applying Lemma~\ref{lem:distance_factorization} and dropping the non-negative second squared term yields the lower bound:
\begin{align*} 
    D(\alpha, \gamma) &\ge \sin(\alpha/2) - \cos(\gamma/2)\sin\left(\frac{\gamma-\alpha}{2}\right) \\
    &=\sin(\alpha/2) + \cos(\gamma/2)\sin\left(\frac{\alpha-\gamma}{2}\right).
\end{align*}

The expected cost for agent $A$ under the reported profile $\vec{x'}$ is given by:
\[\E[d(\mech(\vec{x'}),A)] = \lambda(\gamma)\|A-B\| + \lambda(\gamma)\|A-A'\| + (1-2\lambda(\gamma))D(\alpha, \gamma).\] 

Substituting the respective distances and our lower bound for $D(\alpha, \gamma)$ gives:
\begin{align*}
\E[d(\mech(\vec{x'},A))] &\ge \frac{\cos(\gamma/2)\left[\sin(\alpha/2) + \sin\left(\frac{\alpha-\gamma}{2}\right)\right] + \sin(\alpha/2) + \cos(\gamma/2)\sin\left(\frac{\alpha-\gamma}{2}\right)}{1+\cos(\gamma/2)} \\
&= \sin(\alpha/2) + \frac{2\cos(\gamma/2)\sin\left(\frac{\alpha-\gamma}{2}\right)}{1+\cos(\gamma/2)} > \sin(\alpha/2),
\end{align*}
where the strict inequality relies on our assumption $\gamma < \alpha$.
\end{proof}

\paragraph*{Case 3: Cross-Boundary Deviations}

Finally, it remains to consider the deviations where an extreme agent misreports and both endpoints of the minimum spanning arc change, while the total spanning angle remains less than $\pi$.

\begin{lemma}
    \label{lem:past_stationary}
    Let $\vec{x}$ be a profile with a minimum spanning arc $[0,\alpha]$ for $\alpha < \pi$. Let $\vec{x'} = (x_i',\vec{x}_{-i})$ be such that the minimum spanning arc is $[-\gamma,\beta]$ with $\beta \in [0,\alpha)$,  $\gamma > 0$, and $\gamma+\beta < \pi$. Then $\E[d(\mech(\vec{x}),x_i)] < \E[d(\mech(\vec{x'}),x_i)]$.
\end{lemma}

\begin{proof}
    Without loss of generality let the deviating agent be agent $A$ corresponding to a true location at angle $\alpha$.

     We can decompose this deviation into two steps via an intermediate span $[-\gamma, \alpha]$:
        \begin{enumerate}
            \item \textbf{Expansion Step:} First, consider the arc expanding from the truthful $[0, \alpha]$ to $[-\gamma, \alpha]$. By Lemma~\ref{lem:expanded:arc}, expanding the minimal spanning arc increases the expected cost for all agents located within the original arc.
            
            \item \textbf{Shrinking Step:} Second, consider the arc shrinking from $[-\gamma, \alpha]$ down to $[-\gamma, \beta]$ (since $A$'s true location $\alpha$ is no longer reported). As we prove in Lemma~\ref{lem:shrinking_bound}, when a spanning arc shrinks away from an extreme agent's true location, that agent's expected cost increases.
        \end{enumerate}
        Because both intermediate operations penalize agent $A$, the combined deviation to $-\gamma$ increases $A$'s expected cost.
\end{proof}

Together, Lemmas \ref{lem:expanded:arc}, \ref{lem:removing:center}, \ref{lem:shrinking_bound}, and \ref{lem:past_stationary} cover all possible individual deviations, completing the proof of Theorem \ref{thm:strategyproof}.

\subsubsection{Group-Strategyproofness}

We complete the proof of Theorem \ref{thm:circle:main} by proving that our mechanism is group-strategyproof in expectation. 

\begin{theorem} \label{thm:group_sp}
    The Chord-Midpoint Mechanism $\mech$ is group-strategyproof in expectation.
\end{theorem}

\begin{proof}
    Let the truthful profile $\vec{x}$ span a minimal arc $W$, and a coalition $S$ deviate to a reported profile $\vec{x'}$ spanning $W'$. To prove group-strategyproofness, we must show that if $W \neq W'$, at least one agent in $S$ increases their expected cost. By the logic of Lemma \ref{lem:expanded:arc} and Lemma \ref{lem:removing:center}, any deviation that transitions the spanning arc between the $<\pi$ and the $\ge\pi$ cases, in either direction, increases the expected cost for at least one deviating agent. 
    
    Therefore, we restrict our analysis to the case where both $W$ and $W'$ have spanning angles smaller than $\pi$. Without loss of generality, let $W = [0, \alpha]$ with an agent $B$ at angle $0$ and an agent $A$ at angle $\alpha$. Let $z$ be the midpoint of the arc connecting $A$ and $B$.
    
    We partition the analysis based on how many of the extreme agents belong to the coalition $S$:
    
  \paragraph*{Case 1: Both extreme agents are in $S$.} 
Under truthful reporting, the mechanism $\mech(\vec{x})$ restricts its outputs to the set $\{A, B, z\}$, where $z$ is the chord midpoint of the extremes $A$ and $B$. The expected distance from each extreme agent to the facility is given by:
\[
\E[d(\mech(\vec{x}), A)] = \E[d(\mech(\vec{x}), B)] = \sin(\alpha/2).
\]
Notice that, by the triangle inequality, for any point $y \in \mathbb{R}^2$ we have:
\[
\|A - y\| + \|B - y\| \geq \|A - B\| = 2\sin(\alpha/2).
\]
By linearity of expectation, for any random variable $Y$, the sum of the expected costs is also lower bounded by $2\sin(\alpha/2)$:
\begin{equation*}
    \E[d(Y, A)] + \E[d(Y, B)] \geq 2\sin(\alpha/2).
\end{equation*}

Now consider a manipulated profile $\vec{x}'$ with a modified spanning arc $W' \neq W$, with extreme agents $A'$ and $B'$ and midpoint $z'$ between them. This implies that at least one of the new candidate facility locations in $\{A', B', z'\}$ must lie off the line segment connecting the extremes $A$ and $B$. Consequently, the mechanism's probability support for $\vec{x}'$ must place a strictly positive probability $p > 0$ on at least one realization $y^*$ that does not lie on the original segment between $A$ and $B$.

By the strict triangle inequality, this specific point $y^*$ satisfies:
\[
\|A - y^*\| + \|B - y^*\| > 2\sin(\alpha/2),
\]
while all other possible outputs $y \in \{A', B', z'\}$ still satisfy the lower bound $\|A - y\| + \|B - y\| \geq 2\sin(\alpha/2)$. Taking the expectation over the entire distribution we obtain:
\[
\E[d(\mech(\vec{x}'), A)] + \E[d(\mech(\vec{x}'), B)] > 2\sin(\alpha/2).
\]
By the pigeonhole principle, it follows that the bounds $\E[d(\mech(\vec{x}'), A)] \le \sin(\alpha/2)$ and $\E[d(\mech(\vec{x}'), B)] \le \sin(\alpha/2)$ cannot both hold simultaneously. At least one of the two deviating extreme agents must incur an expected cost greater than $\sin(\alpha/2)$, which ensures the deviation fails.

\paragraph*{Case 2: At least one extreme agent is not in $S$.}
We formalize the deviation of $S$ in two sequential steps by constructing an intermediate profile $\vec{y}$. In profile $\vec{y}$, only the deviating agents located in the interior of $W$ report their altered locations from $\vec{x}'$, while all other agents-including any deviating extreme agents-report truthfully according to $\vec{x}$. Because the extreme agents continue to anchor the boundaries, this interior deviation gives an intermediate spanning arc $W_y$ that either expands or remains equal to $W$. 

By Lemma \ref{lem:expanded:arc}, if this intermediate arc expands, the expected cost for every agent in the interval $[0, \alpha]$ increases. If the coalition $S$ contains no extreme agents, then the intermediate profile constitutes the final reported profile ($\vec{y} = \vec{x}'$). The resulting expansion means that every agent in $S$ incurs a higher expected cost, so the deviation fails.

Alternatively, if exactly one extreme agent is in $S$, we assume without loss of generality that it is agent $A$. We then evaluate the final deviation from the intermediate profile $\vec{y}$ based on the intermediate arc $W_y$:
\begin{itemize}
    \item \textbf{If $W_y$ expands:} The transition from $\vec{x}$ to $\vec{y}$ increases the expected cost for agent $A$, meaning $\E[d(\mech(\vec{y}), A)] > \E[d(\mech(\vec{x}), A)]$. In the subsequent step from $\vec{y}$ to $\vec{x}'$, agent $A$ unilaterally deviates while all other agents maintain their reports from $\vec{y}$. Because agent $A$ reports truthfully in $\vec{y}$, individual strategyproofness (Theorem \ref{thm:strategyproof}) implies that this unilateral deviation cannot decrease their expected cost, ensuring $\E[d(\mech(\vec{x}'), A)] \geq \E[d(\mech(\vec{y}), A)]$. Combining these inequalities gives $\E[d(\mech(\vec{x}'), A)] > \E[d(\mech(\vec{x}), A)]$.
    
    \item \textbf{If $W_y = W$:} The intermediate transition preserves the expected cost for agent $A$, so $\E[d(\mech(\vec{y}), A)] = \E[d(\mech(\vec{x}), A)]$. However, because the final configuration satisfies $W' \neq W$, the step from $\vec{y}$ to $\vec{x}'$ represents a unilateral deviation by agent $A$ that alters the minimal spanning arc. As established by our boundary lemmas (Lemmas \ref{lem:expanded:arc} and \ref{lem:shrinking_bound}), any unilateral deviation by a boundary agent that changes the minimal spanning arc increases that agent's expected cost, giving $\E[d(\mech(\vec{x}'), A)] > \E[d(\mech(\vec{y}), A)]$.
\end{itemize}
Thus, in both subcases, the deviating extreme agent does not profit, ensuring that any coalition is never profitable.
\end{proof}

We complement this section with a separation result, showing that randomization is necessary for group-strategyproof mechanisms to achieve an approximation ratio less than $2$. Our proof adapts techniques from~\cite[Lemma 3.5]{tang_characterization_2020} and~\cite{procaccia_approximate_2013} to our setting.

\begin{restatable}{theorem}{thmGSPLower}
\label{thm:gsp:lower-bound}
Any deterministic, group-strategyproof, and unanimous mechanism $\mech$ for the 2-agent facility location problem on the unit circle $S^1$ with output augmentation has an approximation ratio of at least $2$.
\end{restatable}

\begin{proof}
    Suppose for contradiction there exists a GSP and unanimous mechanism $\mech$ with an approximation ratio of $2 - \epsilon$ for some $\epsilon > 0$.
    
    Let $x_1, x_2 \in S^1$ be two agent locations separated by an angle $\theta \in (0, \pi]$. By GSP and unanimity, the output $y = \mech(x_1, x_2)$ is constrained to the lens-shaped region $L(x_1, x_2)$ bounded by the circle's arc connecting $x_1$ and $x_2$ and its reflection across the chord between them (see Figure \ref{fig:lens_proof}). Otherwise, there would exist a point $z \in S^1$ closer to both agents than $y$. Since the mechanism is unanimous, the agents could profitably deviate by jointly reporting $z$, forcing the output to be $z$ and violating group-strategyproofness.
    
    Without loss of generality, assume $\|y - x_2\| \le \|y - x_1\|$. Let $x_2'$ be the radial projection of $y$ onto the circle's arc connecting $x_1$ and $x_2$, and let $\delta = \|y - x_2'\|$. 
    
   The distance from a point $y$ inside the lens to its projection $x_2'$ is at most twice the maximum distance from the chord to the arc. This maximum distance equals $h = 1 - d(O, m) = 1 - \cos(\theta/2)$, where $O$ is the center of the circle and $m$ is the midpoint of the chord. Using the identity $\cos(2x) = 1 - 2\sin^2(x)$, this can be rewritten as $2\sin^2(\theta/4)$, implying the upper bound: 
    \begin{equation}\label{eq:delta_upper_bound}
        \delta \le 2(1 - \cos(\theta/2)) = 4\sin^2(\theta/4).
    \end{equation}
    
    Fix $\theta > 0$ sufficiently small such that $4\sin(\theta/4) < \epsilon$. 
    
    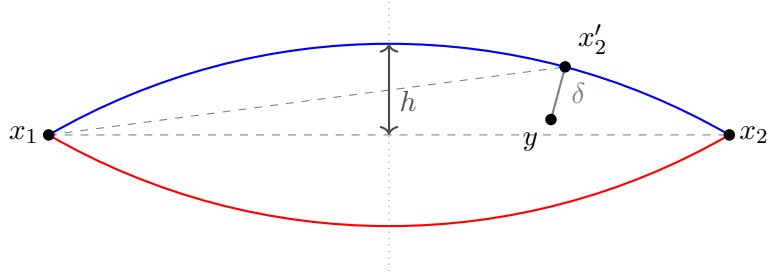
\begin{figure}[htpb]
    \centering
    \begin{tikzpicture}[scale=1.8]
        \coordinate (x1)  at (-2.5, 0);
        \coordinate (x2)  at (2.5, 0);
        
        \coordinate (x2p) at (1.294, 0.500); 
        \coordinate (y)   at (1.190, 0.113); 
    
        \draw[thick, blue] (x2) arc (60:120:5);
        \draw[thick, red]  (x1) arc (240:300:5);
    
        \draw[dashed, gray] (x1) -- (x2);
        \draw[dotted, gray] (0, -1.0) -- (0, 1.0);
    
        \draw[gray, dashed] (x1) -- (x2p);
    
        \draw[thick, gray] (y) -- (x2p) node[midway, right=1pt] {$\delta$};
        
        \fill (x1) circle (1.2pt) node[left] {$x_1$};
        \fill (x2) circle (1.2pt) node[right] {$x_2$};
        \fill[black] (y) circle (1.2pt) node[below left=1pt] {$y$};
        \fill[black] (x2p) circle (1.2pt) node[above right=1pt] {$x_2'$};
    
        \draw[<->, thick, black!70] (-0.0, 0) -- (-0.0, 0.67) node[midway, right, yshift=-4.5pt] {$h$};
    
    \end{tikzpicture}
    \caption{The lens region $L(x_1, x_2)$ constraining the output $y$. The blue arc represents the circle's boundary arc between $x_1$ and $x_2$, and the red arc is its reflection across the chord. The height $h$ denotes the maximum distance from the chord to the boundary arc projection.}
    \label{fig:lens_proof}
    \end{figure}
    
    Consider the profile $(x_1, x_2')$, and let $y' = \mech(x_1, x_2')$. By strategyproofness, an agent at $x_2'$ must not benefit by falsely reporting $x_2$:
\begin{equation}\label{eq:sp_bound}
    \|y' - x_2'\| \le \|\mech(x_1, x_2) - x_2'\| = \|y - x_2'\| = \delta.
\end{equation}

For the profile $(x_1, x_2')$, the optimal maximum cost is $\opt(x_1, x_2') = \frac{1}{2}\|x_1 - x_2'\|$. By the triangle inequality and \eqref{eq:sp_bound}, the cost to agent 1 is bounded by:
\begin{equation}\label{eq:agent1_bound}
    \|y' - x_1\| \ge \|x_1 - x_2'\| - \|y' - x_2'\| \ge \|x_1 - x_2'\| - \delta.
\end{equation}

The approximation ratio $\rho$ for the profile $(x_1, x_2')$ must satisfy, by using \eqref{eq:agent1_bound}:
\begin{equation}\label{eq:rho_lower_bound}
    \rho \ge \frac{\|y' - x_1\|}{\opt(x_1, x_2')} \ge \frac{\|x_1 - x_2'\| - \delta}{\frac{1}{2}\|x_1 - x_2'\|} = 2 - \frac{2\delta}{\|x_1 - x_2'\|}.
\end{equation}
    
    Because we assumed $\|y - x_2\| \le \|y - x_1\|$, the facility $y$ lies on the $x_2$ side of the perpendicular bisector between $x_1$ and $x_2$. Consequently, its radial projection $x_2'$ also lies on the $x_2$ side of the arc. This means $x_2'$ is further from $x_1$ than the arc's midpoint. The angle between $x_1$ and this midpoint is $\theta/2$, which on the unit circle corresponds to a chord length of $2\sin(\frac{\theta/2}{2}) = 2\sin(\theta/4)$. Therefore, we obtain the lower bound:
    \begin{equation}\label{eq:distance_lower_bound}
        \|x_1 - x_2'\| \ge 2\sin(\theta/4).
    \end{equation}
    
    Substituting the upper bound for $\delta$ from \eqref{eq:delta_upper_bound} and the lower bound for $\|x_1 - x_2'\|$ from \eqref{eq:distance_lower_bound} into inequality \eqref{eq:rho_lower_bound} gives:
    \begin{equation*}
        \rho \ge 2 - \frac{2(4\sin^2(\theta/4))}{2\sin(\theta/4)} = 2 - 4\sin(\theta/4).
    \end{equation*}
    
    By our choice of $\theta$, we have $4\sin(\theta/4) < \epsilon$, which implies $\rho > 2 - \epsilon$. This contradicts the assumption that $\mech$ achieves an approximation ratio of $2 - \epsilon$, establishing the lower bound of $2$.
\end{proof}

\begin{remark}
    The unanimity condition of the theorem can be dropped since any mechanism with a bounded approximation ratio must necessarily be unanimous.  Furthermore, by Lemma \ref{lem:population_extension2}, this lower bound of $2$ for the 2-agent setting automatically generalizes to any arbitrary number of agents $n \ge 2$.
\end{remark}

\section{Conclusion and Future Directions}

In this work, we explored the limits of strategyproof mechanisms for the egalitarian facility location problem under the Euclidean distance. One of the main results established were the improved lower bounds for randomized mechanisms in standard multi-dimensional spaces. A related major open problem is to fully close the gap between these new limits and existing mechanisms. A potential path forward involves formally characterizing the class of randomized strategyproof mechanisms; however, given that complete characterizations remain elusive even for the line case, this poses a significant challenge. 


We also introduced a dimension-augmented framework, demonstrating that expanding the allowable output space beyond the agents' input domain can fundamentally increase the optimization capabilities of strategyproof mechanisms and bypass classical barriers. This new paradigm naturally presents parallel open questions, making the closure of the remaining approximation bounds in these augmented settings a logical next step. Notably, while we established a tight bound for deterministic mechanisms in the augmented line setting, it remains an open question whether the introduction of randomization can outperform this $\sqrt{2}$ performance barrier. Furthermore, complete characterizations of deterministic mechanisms remain open for both of the augmented cases we investigated. In particular, for the circle setting, it remains an open question whether there exists any deterministic anonymous and unanimous mechanism that genuinely exploits the output augmentation, rather than merely reducing to a standard strategyproof mechanism over the full $\mathbb{R}^2$ plane via the characterization of Peters et al.~\cite{peters_range_1993}. Therefore, a compelling direction for future work would be to pursue these general characterizations for both deterministic and randomized mechanisms. Additionally, while our analysis focused on specific geometries, future research could explore different pairs of input and output spaces to further map the theoretical landscape of dimension-augmented facility location.

\bibliographystyle{plain} 
\bibliography{references} 

\newpage
\section{Appendix}

\begin{lemma}\label{lem:boundary-points}
Let $S^1$ be a circle in $\mathbb{R}^2$ centered at a point $x_0$ with radius $r$, and let
$k$ points be placed equidistantly on $S^1$ at angles $\theta_i = 2\pi i/k$ for
$i = 0, \ldots, k-1$; call these locations \emph{boundary points}. 
Let $y$ be an arbitrary point. Define $y'$ as the point on $S^1$ that we obtain by extending the line segment from $y$ through the center $x_0$ to the furthest boundary, i.e., 
$y' = x_0 - r \cdot (y - x_0)/\|y - x_0\|$, and let $x'$ be the boundary point closest
to $y'$. Then:
\begin{equation}
    d(y, x') \geq \sqrt{d(y,x_0)^2 + r^2 + 2\, d(y,x_0)\, r \cos\!\left(\frac{\pi}{k}\right)}.
\end{equation}
\end{lemma}

\begin{proof}
Without loss of generality, translate every point so that $x_0 = O$ is the origin. Let
$y = \rho(\cos\alpha, \sin\alpha)$ for some angle $\alpha \in [0, 2\pi)$, where
$\rho = d(y, x_0) = \|y\|> 0$. Then $y' = r \cdot (-\cos\alpha, -\sin\alpha)$, the point
on $S^1$ in the direction from $y$ through $x_0$.

The $k$ boundary points are at angles $\theta_i = 2\pi i/k$. The closest boundary point
$x'$ to $y'$ is at angle $\theta_{i^*}$, where $\theta_{i^*}$ is the angle nearest
to $\alpha + \pi$. The angular deviation $\delta = \theta_{i^*} - (\alpha + \pi)$ satisfies
$|\delta| \leq \pi/k$. Thus:
\[
    x' = r(-\cos(\alpha + \delta), -\sin(\alpha + \delta)) \quad \text{for some $|\delta| \leq \pi/k$. 
}
\]
Expanding the squared Euclidean distance between $y$ and $x'$ we obtain:
\begin{align*}
d(y, x')^2 &= \|\rho \cdot(\cos\alpha, \sin\alpha) + r \cdot(\cos(\alpha+\delta), \sin(\alpha+\delta))\|^2 = \rho^2 + r^2 + 2\rho r\cos\delta.
\end{align*}
Since $\cos(.)$ is decreasing on $[0, \pi]$ and $|\delta| \leq \pi/k$, we have
$\cos\delta \geq \cos(\pi/k)$, and thus:
\begin{equation*}
    d(y, x')^2 \geq \rho^2 + r^2 + 2\rho r\cos\!\left(\frac{\pi}{k}\right) = d(y,x_0)^2 + r^2 + 2\,d(y,x_0)\,r\cos\!\left(\frac{\pi}{k}\right).
\end{equation*}
Taking the square root completes the proof.

\end{proof}

\LBdiscretized*

\begin{proof}
We follow the same line of arguments as in the proof of Theorem~\ref{lowerboundtheorem}. The key difference is that we cannot simply lower bound the maximum distance from the facility location $y'$ to any point on the sphere $C_{x_0}$ by $d(y', x_0) + a_2$ since we only have a finite number of agents deviating from $x_0$, meaning they cannot densely cover the circumference. 

Suppose we distribute the $k = n/3$ agents, initially located at $x_0$, equidistantly over the circle $C_{x_0} = S^1$ as defined in Lemma~\ref{lem:boundary-points}. Let the resulting location profile be $\vec{x}'$. Consider any realized facility location $y'$. Using Lemma~\ref{lem:boundary-points} with $y = y'$ and $r = a_2 = \sqrt{3}$, there exists a boundary point $x'$ on $C_{x_0}$ such that 
\[
d(y', x') \geq \sqrt{d(y',x_0)^2 + 3 + 2\sqrt{3} \, d(y',x_0)\, \cos\!\left(\frac{\pi}{k}\right)}.
\]
The expression on the right-hand side is convex in $d(y',x_0)$ for all $k \ge 2$. We can thus apply Jensen's inequality to conclude that:
\begin{align*}
\SC(\mech(\vec{x}'), \vec{x}') 
& \ge \E\left[\sqrt{d(Y',x_0)^2 + 3 + 2\sqrt{3} \, d(Y',x_0)\, \cos\!\left(\frac{\pi}{k}\right)} \; \right] \\
& \ge \sqrt{\E[d(Y',x_0)]^2 + 3 + 2\sqrt{3} \, \E[d(Y',x_0)]\, \cos\!\left(\frac{\pi}{k}\right)} \\
& \ge \sqrt{4 + 2\sqrt{3} \, \cos\!\left(\frac{\pi}{k}\right)},
\end{align*}
where the final inequality holds because $\E[d(Y',x_0)] \ge 1$ and the underlying function is increasing over this domain. The claim follows by dividing this final expression by the optimal cost $\opt(\vec{x}') = \sqrt{3}$.
\end{proof}

The proof of Theorem \ref{thm:d-dim-lowerbound} relies on the following technical lemma that bounds the maximum distance from a realized facility to the described distribution of points on a hypersphere and generalizes Lemma \ref{lem:boundary-points}.

\begin{lemma}\label{lem:boundary-points-d}
Let $S^{d-1}$ be a $(d-1)$-sphere in $\mathbb{R}^d$ centered at $x_0$ with radius $r$. Let $k \ge (\sqrt{d-1})^{d-1}$ points be placed on $S^{d-1}$ using a spherical coordinate grid with $d-1$ angular coordinates, where each coordinate is discretized into $m = k^{1/(d-1)}$ equidistant values.
Let $y$ be an arbitrary point in $\mathbb{R}^d$. Define $y'$ as the projection of $y$ through $x_0$ to the far boundary of $S^{d-1}$, i.e., $y' = x_0 - r \cdot (y - x_0)/\|y - x_0\|$. Let $x'$ be the grid point on $S^{d-1}$ closest to $y'$. Then:
\begin{equation}
    d(y, x') \geq \sqrt{d(y,x_0)^2 + r^2 + 2\, d(y,x_0)\, r \cos(\gamma)},
\end{equation}
where $\gamma = \frac{\pi\sqrt{d-1}}{k^{1/(d-1)}}$.
\end{lemma}

\begin{proof}
Without loss of generality, translate every point so that $x_0 = O$ is now the origin. Let $y = \rho u$ with $\rho = d(y, x_0) = \|y\| > 0$ and $\|u\| = 1$. The projection is $y' = -ru$. Let the closest grid point be $x' = rv$ with $\|v\| = 1$.

The grid divides each angular axis into $m = k^{1/(d-1)}$ equidistant intervals. Therefore, in each of the $d-1$ angular coordinates, the distance between the coordinates of $y'$ and the closest grid point $x'$ is at most $\frac{\pi}{m}$.

We can bound the angle $\delta$ that $-u$ and $v$ make with the origin by the Euclidean distance of these coordinate differences. Summing over the $d-1$ coordinates, we have:
\[
    \delta \le \sqrt{ \sum_{i=1}^{d-1} \left(\frac{\pi}{m}\right)^2 } = \frac{\pi\sqrt{d-1}}{m} = \gamma.
\]

Expanding the squared Euclidean distance between $y$ and $x'$ we obtain:
\begin{equation*}
    d(y, x')^2 = \|\rho u - r v\|^2 
                = \rho^2 + r^2 + 2\rho r \cos\delta
\end{equation*}
By our assumption that $k \ge (\sqrt{d-1})^{d-1}$, we obtain $\gamma \le \pi$. Since $\cos(\cdot)$ is decreasing on $[0, \pi]$ and we established $\delta \le \gamma \le \pi$, it follows that $\cos\delta \ge \cos\gamma$ and thus:
\[
    d(y, x')^2 \ge d(y, x_0)^2 + r^2 + 2 d(y, x_0) r \cos\gamma.
\]

Taking the square root completes the proof.
\end{proof}

\LBddim*

\begin{proof}
The proof follows the argument of the $\R^2$ case, substituting the circle discretization with the spherical coordinate grid. We distribute $k = n/(d+1)$ agents over the sphere of radius $r = \sqrt{\frac{2(d+1)}{d}}$ centered at $x_0$.

Let $y'$ be the realized facility location. By Lemma \ref{lem:boundary-points-d}, there exists a grid point $x'$ such that:
\[
d(y', x') \geq \sqrt{d(y',x_0)^2 + r^2 + 2r \, d(y',x_0)\, \cos\gamma},
\]
where $\gamma = \frac{\pi\sqrt{d-1}}{(n/(d+1))^{1/(d-1)}}$. 

The function $f(z) = \sqrt{z^2 + r^2 + 2rz\cos\gamma}$ is convex for $z \ge 0$. Applying Jensen's Inequality to the expected social cost $\SC(\mech(\vec{x}'), \vec{x}')$:
\begin{align*}
\SC(\mech(\vec{x}'), \vec{x}') 
&\ge \E\left[\sqrt{d(Y',x_0)^2 + r^2 + 2r \, d(Y',x_0)\, \cos\gamma} \right] \\
&\ge \sqrt{\E[d(Y',x_0)]^2 + r^2 + 2r \, \E[d(Y',x_0)]\, \cos\gamma}.
\end{align*}
Using $\E[d(Y',x_0)] \ge 1$ and substituting $r = \sqrt{\frac{2(d+1)}{d}}$, the maximum cost is lower bounded by:
\[
\sqrt{1 + \frac{2(d+1)}{d} + 2\sqrt{\frac{2(d+1)}{d}} \cos\gamma}.
\]
The theorem follows by dividing by the optimal cost $\opt(\vec{x}') = r = \sqrt{\frac{2(d+1)}{d}}$.
\end{proof}

\MinCostLemma*

\begin{proof}
    We parameterize the facility location as $y = (u, 1+v)$ subject to the distance constraint $u^2 + v^2 = r^2$. The squared distances $D_i$ to the locations in $\vec{x}'$ simplify to:
    \begin{align*}
        D_1 &= d(x_1', y)^2= u^2 + (v - \sqrt{3})^2 = r^2 - 2\sqrt{3}v + 3, \\
        D_2 &= d(x_2, y)^2=  (u - \tfrac{\sqrt{3}}{2})^2 + (v + 1.5)^2 = r^2 - \sqrt{3}u + 3v + 3, \\
        D_3 &= d(x_3, y)^2=  (u + \tfrac{\sqrt{3}}{2})^2 + (v + 1.5)^2 = r^2 + \sqrt{3}u + 3v + 3.
    \end{align*}

    By symmetry, we may assume $u \ge 0$, ensuring $D_3 \ge D_2$. The egalitarian cost is defined by $\max(\sqrt{D_1}, \sqrt{D_3})$, which is minimized when $D_1 = D_3$. Setting $D_1 = D_3$ gives:
    $$ -2\sqrt{3}v = \sqrt{3}u + 3v \implies \sqrt{3}u = -(2\sqrt{3}+3)v. $$
    Since $u \ge 0$, we must have $v \le 0$. Substituting $u^2 = r^2 - v^2$ and squaring both sides gives $3(r^2 - v^2) = (21 + 12\sqrt{3})v^2$, which simplifies to $3r^2 = (24 + 12\sqrt{3})v^2$. Thus, the minimizer is $v = -r / \sqrt{8+4\sqrt{3}}$.

    Substituting $v$ into $D_1$ we obtain the minimal squared cost $L(r)^2$:
    $$ L(r)^2 = r^2 - 2\sqrt{3}v + 3 = r^2 + \left( \frac{2\sqrt{3}}{\sqrt{8+4\sqrt{3}}} \right) r + 3 = r^2 + \lambda r + 3. $$
    Because $\lambda > 0$, the function $L(r) = \sqrt{r^2 + \lambda r + 3}$ is increasing for $r \ge 0$. Also, its second derivative is positive, establishing strict convexity for $r\geq0$.
    
\end{proof}

\begin{restatable}{lemma}{thmPopulationExtension}
\label{lem:population_extension2}
Let $\lambda$ be a lower bound on the approximation ratio of randomized strategyproof 2-agent mechanisms. Then $\lambda$ is also a lower bound for all mechanisms with $n > 2$ agents.
\end{restatable}

\begin{proof}
   Suppose for contradiction there exists a randomized strategyproof $n$-agent mechanism $\mech$ with an approximation ratio $\lambda' < \lambda$. We construct a 2-agent mechanism $\mech'$ that on input $(x_1, x_2)$ runs the $n$-agent mechanism on the profile with the first agent at $x_1$ and the remaining $n-1$ agents co-located at $x_2$:
   
   $$ \mech'(x_1,x_2) = \mech(x_1,x_2,\dots, x_2). $$Because $\mech'$ evaluates $\mech$ on a restricted subdomain, the approximation ratio of $\mech'$ is at most $\lambda'$.

    We show that $\mech'$ inherits strategyproofness from $\mech$ by examining potential deviations by either agent.
    \begin{itemize}
        \item \textbf{Case 1: The first agent deviates.} A deviation by the first agent from $x_1$ to $x_1'$ in $\mech'$ corresponds to a single, unilateral deviation in $\mech$. By the strategyproofness of $\mech$, this agent cannot decrease their expected distance to the facility:
        $$ \E[d(\mech'(x_1', x_2), x_1)] \ge \E[d(\mech'(x_1,x_2), x_1)]. $$
        
        \item \textbf{Case 2: The second agent deviates.} A deviation by the second agent from $x_2$ to $x_2'$ in $\mech'$ corresponds to a simultaneous deviation by a cluster of $n-1$ co-located agents in $\mech$ from their shared true location $x_2$ to $x_2'$. By \Cref{lem_cluster_deviation}, this joint deviation cannot decrease the expected distance to their true location $b$. Thus, 
        $$ \E[d(\mech'(x_1, x_2'), x_2)] \ge \E[d(\mech'(x_1,x_2), x_2)], $$
        meaning the second agent cannot profit.
    \end{itemize}
    Since neither agent has a profitable deviation, the randomized mechanism $\mech'$ is strategyproof. This contradicts the assumption that $\lambda$ was a lower bound for the approximation ratio of 2-agent mechanisms, completing the proof.

\end{proof}

\begin{lemma} \label{lem:right_border_bound}
For any translation and scale invariant, unanimous randomized mechanism $\mech$ with a $y$-axis symmetric output distribution $\mathcal{D}$, under the true profile $\vec{x} = (x_1, x_2) = (-1, 1)$, if Agent~1 reports any location $x_1' \ge 1$, their expected distance to the output is at least $2$.
\end{lemma}

\begin{proof}
By unanimity, if Agent~1 reports $x_1'=1$, both agents report the same location $1$, forcing the mechanism to output the facility at $(1,0)$. The true distance from Agent~1's location $x_1 = (-1,0)$ to this facility makes the expected distance $2$.

For any misreport $x_1' \ge 1$, the sorted reported profile is $\{1, x_1'\}$. By the translation and scale invariance, an arbitrary reference point $(a,b) \sim \mathcal{D}$ is mapped to the facility location $T_{1, x_1'}(a,b)$ given by:
\begin{equation*}
T_{1, x_1'}(a,b) = \left( \frac{x_1'-1}{2}a + \frac{x_1'+1}{2}, \; \frac{x_1'-1}{2}b \right).
\end{equation*}
The function $f_{(a,b)}(x_1')$ representing the distance from Agent~1's true location $x_1 = (-1,0)$ to this mapped facility is $d(x_1, T_{1, x_1'}(a,b))$. Expanding, we obtain:
\begin{equation} \label{eq:f_u_def}
f_{(a,b)}(x_1') = \sqrt{\left( \frac{x_1'-1}{2}a + \frac{x_1'+3}{2} \right)^2 + \left( \frac{x_1'-1}{2}b \right)^2},
\end{equation}
and its derivative with respect to $x_1'$ is given by:
\begin{equation*}
f'_{(a,b)}(x_1') = \frac{\left( \frac{x_1'-1}{2}a + \frac{x_1'+3}{2} \right)\left( \frac{a+1}{2} \right) + \left( \frac{x_1'-1}{2}b \right)\left( \frac{b}{2} \right)}{\sqrt{\left( \frac{x_1'-1}{2}a + \frac{x_1'+3}{2} \right)^2 + \left( \frac{x_1'-1}{2}b \right)^2}}.
\end{equation*}

Evaluating this derivative at the boundary $x_1' = 1$, we obtain:
\begin{equation*} 
f'_{(a,b)}(1) = \frac{\left( 0 + \frac{4}{2} \right)\left( \frac{a+1}{2} \right) + (0)\left( \frac{b}{2} \right)}{\sqrt{\left( 0 + \frac{4}{2} \right)^2 + 0^2}} = \frac{a+1}{2}.
\end{equation*}
Taking the expectation over $\mathcal{D}$ and applying $y$-axis symmetry ($\E_{(a,b)\sim\mathcal{D}}[a] = 0$), the derivative of the expected distance to $x_1$ at the boundary $x_1'=1$ evaluates to:
\begin{equation*} 
\E_{(a,b)\sim\mathcal{D}}\left[ f'_{(a,b)}(1) \right] = \frac{\E_{(a,b)\sim\mathcal{D}}[a] + 1}{2} = \frac{1}{2} > 0.
\end{equation*}

Because the mapped coordinates are affine functions of $x_1'$ and the Euclidean norm is convex, the composite function $\E_{(a,b)\sim\mathcal{D}}\left[ f_{(a,b)}(x_1') \right]$ is convex with respect to $x_1'$ on $[1, \infty)$. A convex function with a strictly positive derivative at its left boundary is increasing across its entire domain. Since the expected distance at the boundary $x_1'=1$ is $2$, we conclude that the expected distance remains at least $2$ for all $x_1' \ge 1$.
\end{proof}

\lemDistanceFactorization*
\begin{proof}
    Let $z'$ be the chord midpoint of the arc $[0, \phi]$, and let the agent be at $x_i$. Applying the Law of Cosines to the triangle defined by the three points $x_i$, $z'$, and $O$, we have 
    \begin{align*}
        D(\theta, \phi)^2 &= \|x_i - O\|^2 + \|O - z'\|^2 - 2 \cdot \|x_i - O\| \cdot \|O - z'\| \cdot \cos\left(\theta - \frac{\phi}{2}\right) \\
        &= 1 + \cos^2\left(\frac{\phi}{2}\right) - 2\cos\left(\frac{\phi}{2}\right)\cos\left(\frac{2\theta - \phi}{2}\right).
    \end{align*}
    Expanding the cross-term using the cosine difference formula, noting $\cos\left(\frac{2\theta - \phi}{2}\right) = \cos\left(\frac{\theta}{2} - \frac{\phi-\theta}{2}\right)$, we obtain:
    \[
    \cos\left(\frac{2\theta - \phi}{2}\right) = \cos\left(\frac{\theta}{2}\right)\cos\left(\frac{\phi-\theta}{2}\right) + \sin\left(\frac{\theta}{2}\right)\sin\left(\frac{\phi-\theta}{2}\right).
    \]
    Substituting this into the equation for $D(\theta, \phi)^2$ yields:
    \[
    D(\theta, \phi)^2 = 1 + \cos^2\left(\frac{\phi}{2}\right) - 2\cos\left(\frac{\phi}{2}\right)\left[\cos\left(\frac{\theta}{2}\right)\cos\left(\frac{\phi-\theta}{2}\right) + \sin\left(\frac{\theta}{2}\right)\sin\left(\frac{\phi-\theta}{2}\right)\right].
    \]
    To factor this into perfect squares, we substitute $1 = \sin^2\left(\frac{\theta}{2}\right) + \cos^2\left(\frac{\theta}{2}\right)$ and expand $\cos^2\left(\frac{\phi}{2}\right) = \cos^2\left(\frac{\phi}{2}\right)\left[\sin^2\left(\frac{\phi-\theta}{2}\right) + \cos^2\left(\frac{\phi-\theta}{2}\right)\right]$. Grouping the respective sine and cosine components of $\frac{\theta}{2}$ gives the stated sum of perfect squares.
\end{proof}

\lemInternalIdentities*
\begin{proof}
    Let $v = \frac{\theta}{2}$ and $u = \frac{\phi}{2}$. Applying Lemma \ref{lem:distance_factorization} with angles $2v$ and $2u$ we obtain:
    \begin{align*}
        D(v, u)^2 &= \big(\sin v - \cos u \sin(u-v)\big)^2 + \big(\cos v - \cos u \cos(u-v)\big)^2\\
        &= g(v, u)^2 + \big(\cos v - \cos u \cos(u-v)\big)^2.
    \end{align*}
    The second term of this expression expands to:
    \begin{align*}
    &\big(\cos v - \cos u \cos(u - v)\big)^2\\
    &= \cos^2 v - 2\cos v \cos u \cos(u-v) + \cos^2 u \cos^2(u-v).
    \end{align*}
    Applying the product-to-sum identity, $2\cos u \cos(u-v) = \cos v + \cos(2u-v)$, allows us to simplify the expression to:
    \begin{align*}
    \cos^2 v - \cos v\left[\cos v + \cos(2u-v)\right] + \cos^2 u \cos^2(u-v) = \sin^2 u \sin^2(u-v).
    \end{align*}
    Finally, using that $\sin^2 u = (1 - \cos u)(1 + \cos u)$, the term simplifies to $(1 + \cos u)h(v, u)$.
\end{proof}

\lemInternalCostSimplification*
\begin{proof}
    Let $v = \frac{\theta}{2}$ and $u = \frac{\phi}{2}$. Let $A'$ be the new extreme point at angle $\phi = 2u$. The selection probability of the extremes $B$ and $A'$ is $\lambda(2u) = \frac{\cos(u)}{2(1+\cos(u))}$, while the selection probability of the new chord midpoint $z'$ is:
    \[
    1-2\cdot\lambda(2u) = 1- \frac{\cos(u)}{1+\cos(u)} = \frac{1}{1+\cos(u)}.
    \]
    Direct substitution of these probabilities and the distances to the outputs ($\|x_i-B\| = 2\cdot\sin(v)$ and $\|x_i-A'\| = 2\cdot\sin(u-v)$) into $\E[d(\mech(\vec{x'}),x_i)]$ we obtain:
    \begin{align*}
    \E[d(\mech(\vec{x'}),x_i)] &= \lambda(2u)\cdot\|x_i-B\| + \lambda(2u)\cdot\|x_i-A'\| + (1-2\cdot\lambda(2u))\cdot\|x_i-z'\|\\
    &= \frac{\cos(u)\cdot\sin(v) + \cos(u)\cdot\sin(u-v) + D(v,u)}{1+\cos(u)} \\
    &= \sin(v) + \frac{D(v,u) - \big(\sin(v) - \cos(u)\cdot\sin(u-v)\big)}{1+\cos(u)}\\
    &= \sin\left(\frac{\theta}{2}\right) + \frac{D(\theta, \phi) - g(\theta, \phi)}{1+\cos\left(\frac{\phi}{2}\right)}.
    \end{align*}
\end{proof}

\lemInternalMonotonicity*
\begin{proof}
    To prove the lemma, it suffices to show that the derivative of $w(\phi)$ with respect to $\phi$ is positive for $\phi \in (\theta, \pi)$. Let $u = \frac{\phi}{2}$ and $v = \frac{\theta}{2}$. By the chain rule, we need $\frac{dw}{du} > 0$ for all $u \in (v, \pi/2)$.

    Differentiating $w(u)$ using the quotient rule is inconvenient due to the complexity of the derivatives. Instead, we will find $\frac{dw}{du}$ using the Implicit Function Theorem.
    
    By the definition of $w(u)$, we have:
    \begin{equation}\label{eq:dist_w}
    D(v,u) = w(u)\cdot(1+\cos(u)) + g(v,u).
    \end{equation}
    Lemma \ref{lem:internal_identities} gives the identity:
    \begin{equation}\label{eq:dist_sq}
    D(v,u)^2 - g(v,u)^2 = (1+\cos(u))\cdot h(v,u).
    \end{equation}
    Substituting \eqref{eq:dist_w} into \eqref{eq:dist_sq} and expanding gives:
    \begin{align*}
    w(u)^2\cdot(1+\cos(u))^2 + 2 w(u)\cdot(1+\cos(u))\cdot g(v,u) &= (1+\cos(u))\cdot h(v,u).
    \end{align*}
    Because $u \in (v, \pi/2)$, we know $1+\cos(u) > 0$. Dividing by this factor we obtain:
    \begin{equation}\label{eq:curve_zero}
    w(u)^2\cdot(1+\cos(u)) + 2 w(u)\cdot g(v,u) - h(v,u) = 0.
    \end{equation}
    
    To apply the Implicit Function Theorem, we define the continuously differentiable function $F$:
    \begin{equation}\label{eq:implicit_F}
    F(x, y) := y^2\cdot(1+\cos(x)) + 2y\cdot g(v,x) - h(v,x).
    \end{equation}
    Equation \eqref{eq:curve_zero} implies that $F(u, w(u)) = 0$. We evaluate the partial derivative of $F$ with respect to $y$ at the point $(u, w(u))$:
    \begin{equation}\label{eq:partial_y}
    \frac{\partial F}{\partial y}(u, w(u)) = 2 w(u)\cdot(1+\cos(u)) + 2\cdot g(v,u) = 2D(v,u).
    \end{equation}
    Because the agent and the chord midpoint are distinct points when $u > v \ge 0$, the distance $D(v,u)$ is positive. Since $\frac{\partial F}{\partial y} > 0$, the Implicit Function Theorem guarantees that $w(u)$ is a $\mathcal{C}^1$ function on this domain, and its derivative is:
    \[
    \frac{dw}{du} = -\frac{\frac{\partial F}{\partial x}(u, w(u))}{\frac{\partial F}{\partial y}(u, w(u))}.
    \]
    Because the denominator is positive, the sign of $\frac{dw}{du}$ is determined by the numerator. Evaluating $-\frac{\partial F}{\partial x}$ at $(u, w(u))$ we obtain a quadratic expression in $w(u)$, which we abbreviate as $w$ for readability:
    \begin{equation}\label{eq:partial_u}
    -\frac{\partial F}{\partial x}(u, w(u)) = w^2\cdot\sin(u) - 2 w\cdot\frac{\partial g}{\partial u} + \frac{\partial h}{\partial u}.
    \end{equation}
    
    Computing the partial derivatives of $g(v,u)$ and $h(v,u)$ gives $\frac{\partial g}{\partial u} = -\cos(2 u - v)$ and $\frac{\partial h}{\partial u} = \sin(u)\cdot\sin^2(u-v) + (1-\cos(u))\cdot\sin(2(u-v))$. We evaluate the sign of \eqref{eq:partial_u} over two subcases:
    
    \begin{itemize}
        \item \textbf{Case $\mathbf{\frac{\partial g}{\partial u} \le 0}$:} Since $h(v,u) \ge 0$, equation \eqref{eq:dist_sq} ensures $D(v,u)^2 \ge g(v,u)^2$, meaning $D(v,u) \ge g(v,u)$ and therefore $w \ge 0$. This implies $-2\cdot w\cdot\frac{\partial g}{\partial u} \ge 0$. Because $\sin(u) > 0$ and $\frac{\partial h}{\partial u} > 0$ for $u \in (v, \pi/2)$, the entire expression in \eqref{eq:partial_u} is positive.
        
        \item \textbf{Case $\mathbf{\frac{\partial g}{\partial u} > 0}$:} We analyze the discriminant of \eqref{eq:partial_u}, $\Delta = 4\cdot\left(\frac{\partial g}{\partial u}\right)^2 - 4\cdot\sin(u)\cdot\frac{\partial h}{\partial u}$. Expanding we obtain:
        \begin{equation}\label{eq:discriminant}
        \frac{\Delta}{4} = \cos(u-v)\cdot\left[\cos^2(u)\cdot\cos(u-v) - 2\cdot\sin(u)\cdot\sin(u-v)\right].
        \end{equation}
        The assumption $\frac{\partial g}{\partial u} > 0$ implies $0>\cos(2 u-v) = \cos(u)\cdot\cos(u-v)- \sin(u)\cdot\sin(u-v)$, we obtain $\cos(u)\cdot\cos(u-v) < \sin(u)\cdot\sin(u-v) $. Because $u \in (0, \pi/2)$, multiplying by $\cos(u) \in (0,1)$ and using the bound $\cos(u) < 2$ gives:
        \[
        \cos^2(u)\cdot\cos(u-v) < 2\cdot\sin(u)\cdot\sin(u-v).
        \]
        This ensures the bracketed term in \eqref{eq:discriminant} is negative, meaning $\Delta < 0$. With a positive leading coefficient ($\sin(u) > 0$) and a negative discriminant, the quadratic has no real roots and remains positive for all $w$.
    \end{itemize}
    In both cases, we obtain $\frac{dw}{du} > 0$. Thus, the expected cost function increases with respect to the spanning angle $\phi$.
\end{proof}

\end{document}